\documentclass[final,onefignum,onealgnum,twoside]{siamart171218}
\usepackage{amssymb}
\usepackage{amsmath}
\usepackage[ruled]{algorithm2e} %

\usepackage{exmath}
\usepackage{subcaption}
\usepackage{hyperref}
\usepackage{pgfplots}
\usepackage{pgfplotstable}
\usepackage[square,sort,comma,numbers]{natbib}

\crefname{subsection}{Section}{Sections}
\crefname{section}{Section}{Sections}
\crefname{algocf}{Algorithm}{Algorithms}
\crefname{assumption}{Assumption}{Assumptions}
\usetikzlibrary{matrix}

\pgfplotsset{compat=newest}
\pgfplotsset{yticklabel style={text width=2.5em,align=right}}
\pgfplotsset{xticklabel style={text height=0.7em}}
\pgfplotsset{/pgfplots/error bars/error bar style={thick, solid}}
\pgfplotsset{/pgfplots/error bars/error mark options={line
    width=0.9pt,mark size=3pt,rotate=90}}
\pgfplotsset{
  every axis plot/.append style={thick, black},
  every axis plot/.append style={
    every mark/.append style={mark size=3,solid,fill opacity=0}
  }
}

\pgfplotsset{
  fixed error bar/.style n args={3}{
    error bars/y dir=both,
    error bars/y explicit,
    table/x=#1,
    table/y=#2,
    table/y error plus=#3,
    table/y error minus expr={
      ifthenelse(
      \thisrow{#2} - \thisrow{#3} <= 0,
      \thisrow{#2} - 0.5 * \pgfkeysvalueof{/pgfplots/ymin},
      \thisrow{#3})
    },
  },
}

\newlength\figureheight
\newlength\figurewidth
\hypersetup{
    colorlinks,
    linkcolor={blue!50!black},
    citecolor={green!50!black},
    urlcolor={red!80!black}
}

\SetAlFnt{\small}
\SetAlCapFnt{\small}
\SetAlCapNameFnt{\small}
\SetAlCapHSkip{0pt}
\IncMargin{-\parindent}

\def\hd{\widehat{d}}
\def\hsig{\widehat{\sigma}}
\def\hdelta{\widehat{\delta}}
\def\sig{{\sigma}}
\def\inner{\widehat{\textnormal{E}}}
\def\MCest{\widehat{H}}    %
\def\varf{V^{\textrm{f}}} %
\def\Ef{E^{\textrm{f}}}   %
\def\eps{\varepsilon}     %
\def\hL{{\widehat{L}}}       %
\def\heta{{\widehat\eta}}     %
\def\fN{\mathcal{N}}       %
\def\hfN{\widehat{\mathcal{N}}}      %
\def\logtol2{\ensuremath{\abs{\log\tol}^2}}
\def\ndelta{\nu}

\newcommand{\fracs}[2]{{\textstyle \frac{#1}{#2}}}
\providecommand{\heavisidesymb}{{\ensuremath{ \textnormal{H}}}}
\providecommand{\heaviside}[1]{{\ensuremath{ \heavisidesymb\p*{#1}}}}
\providecommand{\ind}[1]{{\ensuremath{ \boldsymbol{1}_{#1}}}}
\providecommand{\fpow}[2]{{{#1}/{#2}}}
\renewcommand{\th}{{\ensuremath{{\textnormal{th}}}}}
\providecommand{\tol}{\ensuremath{\varepsilon}}

\providecommand{\revise}[1]{#1}

\newsiamremark{assumption}{Assumption}
\newsiamremark{remark}{Remark}

\providecommand\figlabel{}

\newcommand\subfig[1]{
  \begin{subfigure}{0.485\textwidth}\input{imgs/#1}\caption{}\end{subfigure}}

\newcommand{\maxNell}{\max\p{N_\ell, N_{\ell-1}}}

\AtBeginDocument{
  \everydisplay=\expandafter{%
    \the\everydisplay
    \medmuskip=4mu plus 2mu minus 4mu
    \thickmuskip=5mu plus 5mu
  }%
}

\headers{Multilevel nested simulation for risk estimation}{M. B. Giles, A.
  Haji-Ali}

\title{Multilevel nested simulation for efficient risk estimation
\thanks{Submitted 28 February 2018.}}

\author{Michael B. Giles \thanks{University of Oxford
    (\email{mike.giles@maths.ox.ac.uk}, \email{hajiali@maths.ox.ac.uk}).} \and
  Abdul-Lateef Haji-Ali\footnotemark[2]}

\begin{document}
\maketitle

\begin{abstract}
  We investigate the problem of computing a nested expectation of the form \\
  ${\prob{{\E{X \given Y}} \geq 0} = \E{\heaviside{\E{X\given Y}}}}$ where
  $\heavisidesymb$ is the Heaviside function. This nested expectation appears,
  for example, when estimating the probability of a large loss from a financial
  portfolio.
  We present a method that combines the idea of using Multilevel Monte
  Carlo (MLMC) for nested expectations %
  with the idea of adaptively selecting the number of samples in the
  approximation of the inner expectation, as proposed
  by~%
  (Broadie et al., 2011). We propose and analyse an algorithm that
  adaptively selects the number of inner samples on each MLMC level
  and prove that the resulting MLMC method with adaptive sampling has
  an $\Order{\tol^{-2}\logtol2}$ complexity to achieve a root
  mean-squared error $\tol$.
    The theoretical analysis is verified by numerical experiments on a
    simple model problem. We also present a stochastic
      root-finding algorithm that, combined with our adaptive methods,
      can be used to compute other risk measures such as Value-at-Risk
      (VaR) and Conditional Value-at-Risk (CVaR), with the latter
      being achieved with $\Order{\tol^{-2}}$ complexity.
\end{abstract}



\section{Introduction}\label{sec:intro}

Our focus in the current work is on computing the following quantity
of interest
\begin{equation}
    \label{eq:objective-nested}
  \eta = \E{\heaviside{\E{X \given Y}}},
\end{equation}
where $\textnormal{H}$ is the Heaviside step function (i.e.,
$\heaviside{x} = 1$ for $x \geq 0$ and $0$ otherwise) and the
inner expectation of the \revise{one-dimensional} random variable,
$X$, is conditional on the value of the outer
\revise{multi-dimensional} random variable, $Y$. This problem appears
in many settings. However, our main motivation for looking at this
problem is to compute the probability of a large loss from a financial
portfolio, see for example~\cite{gordy:nested, broadie:adapt}. In such
a context, $X$ would be \revise{a one-dimensional random variable
  equal to the} sum of losses at maturity from the options in the
portfolio in excess of some threshold value, while $Y$ would be a
\revise{multi-dimensional random variable that includes the} values of
the underlying stocks and other risk factors at some short risk
horizon.

Our approach to approximating $\eta$ in~\eqref{eq:objective-nested}
draws from the work of Gordy \& Juneja~\cite{gordy:nested}. There, the
authors expressed the probability of a large loss as a nested
expectation
then proposed using nested Monte Carlo samplers to estimate both the
outer and inner expectations. They proved that using
$\Order{\tol^{-1}}$ samples in the inner Monte Carlo sampler and
$\Order{\tol^{-2}}$ samples in the outer Monte Carlo sampler is
sufficient, under certain conditions, to achieve a root mean squared
(RME) error of $\Order{\tol}$ in the estimation of $\eta$. Thus the
total cost of their method is
$\Order{\tol^{-3}}$. See~\cite{giorgi:ml2r} for sharper and extended
analysis of their results. In~\cite{broadie:adapt}, Broadie et.\ al.\
improved the complexity of the nested Monte Carlo method by adapting
the number of samples in the inner Monte Carlo sampler to the specific
sample of the outer random variable, $Y$. The basic idea relies on the
fact that the step function in~\eqref{eq:objective-nested} does not
change value when a large error is committed in the estimation of the
inner conditional expectation, $\E{X \given Y}$, provided it is
sufficiently far from $0$. Hence, depending on $\abs{\E{X \given Y}}$
and the size of the statistical error that is committed when
approximating ${\E{X \given Y}}$ with a Monte Carlo sampler, the
required number of samples is determined with a cap which is again
$\Order{\tol^{-1}}$.  Using this idea, it was shown
in~\cite{broadie:adapt} that, under certain conditions, the expected
required number of samples in the inner Monte Carlo estimator is
$\Order{\tol^{-1/2}}$, bringing the total cost of the nested Monte
Carlo method down to $\Order{\tol^{-5/2}}$.  We will review in more
details the nested Monte Carlo method and the adaptive method to
select the number of inner samples in \cref{sec:nested-mc}.

In parallel, the Multilevel Monte Carlo (MLMC) method was introduced
in~\cite{giles:MLMC} to reduce the complexity of Monte Carlo methods
when only approximate samples of a random variable can be generated
and, hence, Monte Carlo approximations are necessarily biased. Assume
we wish to estimate $\E{g}$ for some random variable, $g$. Denoting
the $\ell^\th$ level approximation of $\E{g}$ by $\E{g_\ell}$ such
that
$\abs*{\E{g_\ell} - \E{g}} =\Order{2^{-\alpha\ell}}
\xrightarrow{\ell \to \infty} 0$ for some $\alpha > 0$, the first
step to approximate $\E{g}$, up to a RMS error $\tol$, is to select
the approximation level, $L$, such that
$\abs*{\E{g} - \E{g_L}} = \Order{\tol}$. A standard Monte Carlo
method then approximates $\E{g_L}$ with a sufficient number of samples
such that the standard deviation of the approximation is also
$\Order{\tol}$. If we assume that the cost of computing a single
sample of $g_\ell$ is $\Order{2^{\gamma \ell}}$, the complexity of
this Monte Carlo method is $\Order{\tol^{-2 - {\gamma}/{\alpha}}}$. On
the other hand, MLMC is based on the telescopic sum
\[ \E{g_L} = \E{g_0} + \sum_{\ell=1}^{L} \E{g_\ell - g_{\ell-1}}.\]
Then, we estimate $\E{g_0}$ and $\E{g_\ell - g_{\ell-1}}$, for
$\ell \leq L$, with independent Monte Carlo samplers with a number of
samples that is decreasing as $\ell \to 0$. Provided that $g_\ell$ and
$g_{\ell-1}$ are sufficiently correlated such that
$\var{g_\ell - g_{\ell-1}} = \Order{2^{-\beta \ell}}$ for
$\beta>0$, the complexity of MLMC is
$\Order{\tol^{-2-\max\p*{0, \p{\gamma-\beta}/{\alpha}}}}$ for
$\gamma \neq \beta$ and $\Order{\tol^{-2}\logtol2}$ for
$\gamma = \beta$~\cite{giles:acta}.
Bujok et.\ al.~\cite{bujok:bernoulli} applied MLMC to nested
  expectations of the form \linebreak ${\E{\phi(\E{X \given Y})}}$ for
  a piece-wise linear function, $\phi$, and $X$ being a Bernoulli
  random variable. Giles~\cite{giles:acta} considers the case in which
  $X$ is a more general random variable with $\phi$ being a twice
  differentiable function. In any case, approximate samples of the
  random variable ${\E{X \given Y}}$ are generated with a Monte Carlo
  sampler with $N_\ell  =  2^\ell$ samples and then a hierarchy of
  approximations based on the number of inner samples is used in an
  MLMC setting. The authors of both papers proved that MLMC in this
  case has an $\Order{\tol^{-2}\logtol2}$ complexity. Moreover, they
  showed that using an antithetic estimator improves the complexity of
  MLMC to $\Order{\tol^{-2}}$. On the other hand, if $\phi$ is a step
  function, as is the case when computing the probability of large
  loss, then we will prove in \cref{s:mlmc-nested} and show
  numerically in~\cref{s:model-example} that MLMC with
  deterministic sampling has $\Order{\tol^{-5/2}}$
  complexity. Moreover, in this case, using antithetic sampling does
  not improve the complexity.

The main contribution of this work is to further reduce the complexity
of the MLMC method for step functions to $\Order{\tol^{-2}\logtol2}$,
by using ideas from~\cite{broadie:adapt} to adapt the number of inner
samples, on each level of the MLMC method, based on the particular
realisation of the random variable, $Y$, of the outer Monte Carlo
sampler.
We start in \cref{sec:nested-est} where we recall in more
detail the work and results of~\cite{gordy:nested, broadie:adapt} on
nested Monte Carlo estimators (see
  also~\cite{hong:var-review} for a more thorough review)
and~\cite{giles:acta} on using MLMC to estimate nested
expectations. Then, in \cref{s:adaptive} we start by motivating
an algorithm, inspired by~\cite{broadie:adapt}, that selects the
number of inner samples given a realisation of the outer random
variable, $Y$. \cref{s:adaptive} also includes an analysis of
the adaptive algorithm to show that a MLMC estimator that uses this
adaptive strategy has levels whose variance of differences converges
with rate $\Order{2^{-\ell}}$ while the work increases with rate
$\Order{2^{\ell}}$ as $\ell$ increases, hence the complexity of MLMC
is $\Order{\tol^{-2}\logtol2}$ in this case.

In \cref{s:model-example}, a simple model problem that mimics
the problem of computing the probability from a large loss from a
financial portfolio is presented. We apply the MLMC method with
adaptive and deterministic sampling and show that our analysis matches
well with numerical experiments. In \cref{s:beyond-prob} we
discuss how our methods can be combined with a simple root-finding
algorithm to compute other risk measures, namely, the Value-at-Risk
(VaR) and the Conditional Value-at-Risk (CVaR). Slightly surprisingly,
we prove in the case of CVaR that the complexity of MLMC with
deterministic sampling is the optimal $\Order{\tol^{-2}}$.  Finally,
conclusions and future work directions are discussed in
\cref{sec:conc}.
\section{Nested estimators for nested
  expectations}\label{sec:nested-est}
\subsection{Nested MC estimators}\label{sec:nested-mc}
In this section we review the works of Gordy \&
Juneja~\cite{gordy:nested} and Broadie et.\ al~\cite{broadie:adapt}
which are based on approximating the inner and outer expectations
in~\eqref{eq:objective-nested} using Monte Carlo.
In~\cite{gordy:nested}, the conditional \emph{inner} expectation
$\E{X \given Y = y}$, for a given $y$, is estimated using an
unbiased Monte Carlo estimator with $N$ samples as follows:
  \begin{equation}\label{eq:inner-estimator-mc}
  \inner_N(y) = \frac{1}{N}\sum_{n=1}^N  x^{(n)}(y),
\end{equation}
where $x^{(n)}(y)$ is the $n$'th sample of the random variable $X$
given $Y = y$ and $\en\{\}{x^{(n)}}_{n}$ are mutually independent,
conditional on $Y = y$.  Then, Monte Carlo is used again to
approximate the \emph{outer} expectation as follows:
  \begin{equation}\label{eq:outer-estimator}
  \eta \approx \frac{1}{M} \sum_{m=1}^M \heaviside{\inner_N(y^{(m)})},
\end{equation}
where $y^{(m)}$ is the $m$'th sample of the random variable $Y$ and
$\en\{\}{y^{(m)}}_{m}$ are mutually independent.
  Denote the joint density of the two random variables $\E{X \given Y}$ and
  $\inner_N(Y)$ by $p_N(y, z)$. Moreover, for $i = 0,1$ and $2$, assume
  that $\frac{\partial^i }{\partial y^i} p_N(y,z)$ exists and that
  there exist a non-negative function, $p_{i,N}$, such that
  \[%
    \begin{aligned}
      \abs*{\frac{\partial^i }{\partial y^i} p_N(y,z)} &\leq p_{i,
        N}(z) &\text{for all $N, y$ and $z$},\\
      \text{and } \quad \sup_N \int_{-\infty}^\infty \abs{z}^q
      p_{i,N}(z)\D z &< \infty, &
    \end{aligned}
  \]
  \revise{for all $0 \leq q \leq 4$}, then, by~\cite[Proposition
  1]{gordy:nested}, the RMS error of the
  estimator~\eqref{eq:outer-estimator} is
  $\Order{M^{-1/2} + N^{-1}}$. Hence, to get an RMS error of
  $\Order{\tol}$, we require $M = \Order{\tol^{-2}}$ and
  ${N = \Order{\tol^{-1}}}$, and the total complexity is
  $\Order{\tol^{-3}}$.

  Note that, given a sample of the random variable $Y$, we want to
  evaluate a step function that is based on the expectation
  $\E{X \given Y}$ which is approximated using a Monte Carlo estimator
  with $N$ samples. However, depending on how far the expectation
  $\E{X \given Y}$ is from zero, where the step function changes
  value, a very rough approximation of ${\E{X \given Y}}$ might still
  give the correct value for the step function $\heaviside{\E{X \given
      Y}}$. This motivates adapting the number of samples based on the
  value of $\abs{\E{X \given Y}}$. Such a method was introduced and
  analysed in~\cite{broadie:adapt} for a nested Monte Carlo
  approximation.

  Heuristically, assuming that given $Y$ we have an estimate
  $\inner_N(Y) > 0$, then the authors in~\cite{broadie:adapt} ask:
  what is the probability that adding an extra sample will produce a
  negative estimate?  Denoting the random variables
  $\mu \eqdef {\E{X \given Y}}$ and $\sigma^2\eqdef\var{X \given Y}$
  and using Chebyshev's inequality yields
\[
  \begin{aligned}
    \prob{\inner_{N+1}(Y) \leq 0\given \inner_{N}(Y)} &= \prob*{ N
      \inner_N(Y) + \mu \leq \mu - x^{(N+1)}(Y)
      \given \inner_{N}(Y)}\\
    &\leq \prob*{ N\inner_N(Y) + \mu \leq \abs*{\mu - x^{(N+1)}(Y)} \given \inner_{N}(Y)} \\
    &\leq \frac{\sigma^2}{\p{N \inner_N(Y) + \mu }^2} \approx
    \frac{\sigma^2}{N^2 \mu^2}.
  \end{aligned}
\]
Then, for $d \eqdef \abs{\mu}$, we only need
$N \geq \tol^{-{1}/{2}} {\sigma}/{{d}}$ in order to ensure that adding
a new inner sample does not change the value of the step function with
probability \revise{$1 - \tol$}, i.e., that our estimate
$\heaviside{\inner_{N}(Y)} = \heaviside{\inner_{N+1}(Y)} \approx
\heaviside{\E{X \given Y}}$ is correct with probability
\revise{$1 - \tol$}. In~\cite{broadie:adapt}, two algorithms were
introduced to adaptively determine the number of samples in the inner
Monte Carlo sampler of a nested Monte Carlo method. The first is based
on minimising the total required number of samples for all inner Monte
Carlo samplers subject to some error tolerance. The second algorithm
is iterative, such that in every iteration estimates of $d$ and $\sig$
given $Y$ are computed using $N$ inner samples, then more inner
samples are added until ${N d}/{\sig}$ exceeds some error margin
threshold. In the current work, we instead start by adding a cap on
the number of samples and set
\begin{equation}\label{eq:mc-adaptive-samples}
  N = \en*\lceil\rceil{\min\p*{\Order{\tol^{-1}},
      \tol^{-{1}/{2}}\frac{\sigma}{{d}}}}.
\end{equation}
Hence when
$\delta \eqdef {{d}}/{\sigma} = \order{\tol^{{1}/{2}}}$, we
would use the maximum number of samples, otherwise, fewer samples are
used and we still get the same estimate for $\heaviside{\inner_N(Y)}$ with
probability \revise{$1 - \tol$}. Assuming that the random variable $\delta$ has a
distribution with a bounded density near $0$, the average number of
samples is then $\Order{\tol^{-{1}/{2}}}$ and the complexity of the
nested Monte Carlo method with this choice of number of samples is
$\Order{\tol^{-{5}/{2}}}$.

\subsection{MLMC for nested expectation}\label{s:mlmc-nested}

In this section, based on ideas in~\cite{bujok:bernoulli, giles:acta,
  hajiali:msthesis}, we will use a MLMC estimator to approximate the
\emph{outer} expectation, using the number of samples in the inner
Monte Carlo as a discretization parameter, as follows
  \begin{equation}\label{eq:mlmc-outer-est}
  \widehat{\eta} \eqdef \sum_{\ell=0}^L \frac{1}{M_\ell}
  \sum_{m=1}^{M_\ell} \heaviside{\inner^{(\mathrm{f}, \ell,
      m)}_{N_\ell}( y^{(\ell,m)})} -
  \heaviside{\inner^{(\mathrm{c},\ell, m)}_{N_{\ell-1}}(
    y^{(\ell,m)})},
\end{equation}
where the estimator of the inner expectation, $\inner_{N_\ell}$, is as
defined in~\eqref{eq:inner-estimator-mc} where, as before,
\begin{equation}\label{eq:mlmc-inner-level}
  \inner^{\p{f,\ell, m}}_{N_\ell}(y) \eqdef \frac{1}{N_\ell}
  \sum_{n=1}^{N_\ell} x^{(f,\ell, m, n)}(y)
\end{equation}
with $\inner^{\p{c,\ell, m}}_{N_{\ell-1}}$ defined similarly but with
$\heaviside{\inner_{N_{-1}}^{\p{\textrm c, 0, \cdots}}(\cdot)} \eqdef
0$. Here, ${\{x^{(\cdot,\ell, m, n)}(y)\}}_{n}$ are i.i.d.\
samples of $X$ given $Y = y$.  Moreover, the samples of the random
variable, $X$, used in $\inner^{\p{\mathrm{f},\ell, m}}_{N_\ell}$ and
$\inner^{\p{\mathrm{c},\ell, m}}_{N_{\ell-1}}$ are mutually
independent, but conditional on the same $y$, and independent from the
samples used in $\inner^{\p{\mathrm{f},\ell', m'}}_{N_\ell}$ and
$ \inner^{\p{\mathrm{c},\ell', m'}}_{N_{\ell-1}}$ for any
$ \ell' \neq \ell$ or $ m' \neq m$.

Again, due to~\cite[Proposition 1]{gordy:nested} and under the
assumptions listed therein, and recalled in
\cref{sec:nested-mc}, we have
  \begin{equation}\label{eq:weak-convergence}
  \abs*{\E*{\heaviside{\inner_{N_{\ell}}( Y)} - \heaviside{\E{X \given
          Y}}}} = \Order{N_{\ell}^{-1}}.
\end{equation}
Moreover, the cost to generate samples of
$\heaviside{\inner_{N_\ell}(\cdot)}$ is $\Order{N_\ell}$. Unlike
standard Monte Carlo, MLMC also requires strong convergence of the
estimators. More specifically, the variance of the difference
$ \heaviside{\inner_{N_\ell}(y)} - \heaviside{\inner_{N_{\ell-1}}(y)}
$ must converge sufficiently fast. At this point, we will prove
another general result regarding nested Monte Carlo samplers. First,
we list the first of two main assumptions in the current work. In what
follows, we define the \revise{one-dimensional} random variables
$d \eqdef \abs*{\E{X \given Y} }$, $\sig^2 \eqdef {\var{X \given Y} }$
and $\delta \eqdef {d}/{\sigma}$.

\begin{assumption}\label{ass:bounded-density}
  We assume that the probability density function of the non-negative,
  \revise{one-dimensional} random variable, $\delta$, denoted by
  $\rho$, exists. Moreover, we assume that there exist constants
  $\rho_{0} > 0$ and $\delta_0 > 0$ such that
  $\rho(\delta) \leq \rho_{0}$ for all
  $\delta \in \en*[]{0, \delta_0}$.
\end{assumption}

Based on this assumption and for a positive, non-increasing
  function $a$, we have
  \begin{equation}\label{eq:int-delta-identity}
  \int_0^\infty {a(\delta)} \rho(\delta) \D \delta \leq \rho_0 \int_{0}^{\infty} a(\delta) \D \delta +
  a(\delta_0),
\end{equation}
which can be shown by splitting the integral at $\delta_0$ and then
bounding $\rho$ using \cref{ass:bounded-density}. We will
also make repeated use of the following identity for a $q > 1$ and a
constant $b>0$
  \begin{equation}\label{eq:int-min-identity}
  \int_{0}^\infty \min\p*{1, b x^{-q}} \D x = \frac{q\, b^{1/q}}{q-1},
\end{equation}
which can be shown by splitting the integral at $b^{1/q}$.

\begin{proposition}[Variance of the inner MC]\label{thm:inner-mc-var}
  Let $X, Y$ be two random variables satisfying
  \cref{ass:bounded-density}. %
  Then
    \begin{equation}\label{eq:nested-mc-var}
    \aligned \var{\heaviside{\inner_N(Y)} - \heaviside{\E{X \given
          Y}}} &\leq \E*{\p*{\heaviside{\inner_{N}(Y)} \!-\!
        \heaviside{\E{X \given
            Y}}}^2} \\
    &= \Order{N^{-1/2}}. \endaligned
  \end{equation}
\end{proposition}
\begin{proof}
  We start from
\[
  \begin{aligned}
    &\E*{\p*{\heaviside{\inner_{N}(Y) } \!-\! \heaviside{\E{X \given
            Y}}}^2 \given Y} \\
    &\hskip 1cm =\prob[\Big]{\abs*{\heaviside{\inner_{N}(Y) } \!-\! \heaviside{\E{X
            \given Y} }}
      = 1\given Y} \ \\
    &\hskip 1cm \leq \ { \prob[\Big]{\, \abs*{{\inner_{N}(Y)} \!-\!  {\E{X \given
              Y} }} \geq d \ \given Y}},
  \end{aligned}
\]
where, by Chebyshev's inequality and bounding the probability by 1, we have
  \begin{equation}\label{eq:chebyshev-prob-bound}
  \aligned {\prob[\Big]{\, \abs{{\inner_{N}(Y)} \!-\!  {\E{X \given Y}
        }} \geq d \given Y}} &\leq \ \min\p*{1,
    {d^{-2}\var*{\inner_N(Y)
        \given Y}}} \\
  &= \min\p*{1, \delta^{-2} N^{-1}}. \endaligned
\end{equation}
Taking expectation over $Y$ yields
\[
  \begin{aligned}
    \E*{\p*{\heaviside{\inner_{N}(Y)} \!-\! \heaviside{\E{X \given
            Y}}}^2}
    &\leq \int_{0}^\infty \min\p*{1,
      {\delta^{-2} \, N^{-1}
      }} \rho(\delta) \D \delta \\
    &\hskip-2cm \leq \rho_0 \int_{0}^\infty \min\p*{1, {\delta^{-2} \, N^{-1} }} \D
    \delta +
    \min\p*{1, {\delta_0^{-2} \, N^{-1} }}\\
    &\hskip-2cm \leq 2\rho_{0} N^{-1/2} + \delta_0^{-2} \, N^{-1}.
  \end{aligned}
\]
Here we first used~\eqref{eq:int-delta-identity} and
then~\eqref{eq:int-min-identity}.  Hence, the variance is
$\Order{N^{-1/2}}$.
\end{proof}

Based on this proposition and under
\cref{ass:bounded-density}, for a \emph{deterministic}
choice of $N_\ell$, increasing with level, we have
  \begin{equation}\label{eq:var-converence}\aligned
  {\var*{\heaviside{\inner_{N_\ell}(Y)} - \heaviside{\inner_{N_{\ell-1}}(Y)}}} =
  \Order{N_{\ell-1}^{-1/2}}, \endaligned
\end{equation}
since for any two random variables $V$ and $W$, we have $\var{V + W}
\leq 2\var{V} +
2\var{W}$.  Hence, if we make a simple choice of $N_\ell=N_0
2^\ell$ and provided the assumptions of~\cite[Proposition
1]{gordy:nested} are satisfied so that we
have~\eqref{eq:weak-convergence}, we can conclude, according to
standard MLMC complexity analysis (with
$\alpha = 2\beta =\gamma =1$), that the cost to approximate
$\eta$ to an error tolerance of $\tol$ is $\Order{\tol^{-{5}/{2}}}$.

\begin{remark}[Using antithetic sampling]\label{rem:antithetic}
  In some nested simulation applications, it is possible to use an
  antithetic estimator, where all samples in the fine estimator,
  $\inner^{{\p{\textrm f, \cdots}}}$, are used in the coarse
  estimator, $\inner^{{\p{\textrm c, \cdots}}}$, in each level to
  increase the correlation between the fine and coarse estimators of
  the inner conditional expectation; see, for
  example,~\cite{giles:acta} or \cref{s:model-example} for
  more details.  In fact, such an estimator was used to compute the
  moments of the total loss from a portfolio at a risk horizon,
  in~\cite{gou:ms-var}. However, using antithetic sampling does not
  change the variance convergence rate in our setting because the
  discontinuity in the step function violates the differentiability
  requirements of the antithetic estimator. In fact, the main
  correlation between the fine and coarse samples is due to the strong
  convergence of the inner Monte Carlo
  estimator~\eqref{eq:nested-mc-var} and using independent samples is
  sufficient to get the rate
  in~\eqref{eq:var-converence}. Nevertheless, using an antithetic
  estimator might reduce the variance by a constant, as we will
  discuss in \cref{s:model-example}. We do not emphasise this
  in the presented theorems for clarity of presentation.
\end{remark}

\subsection{Adaptive Sampling for MLMC}\label{s:adaptive}
In this section, we apply the adaptive sampling method mentioned in
\cref{sec:nested-mc} to the MLMC estimator of the outer
expectation and analyse the resulting algorithm. Our aim is to
increase the variance convergence to $\Order{2^{-\ell}}$ while still
using $\Order{2^{\ell}}$ inner samples per level on average, so that
the MLMC complexity is $\Order{\tol^{-2}\logtol2}$.  In this section,
in addition to \cref{ass:bounded-density}, we also work
under the following assumption.
\begin{assumption}\label{ass:bounded-moments}
  We assume that
  for some $2 < q < \infty$, the $q^{\textnormal{th}}$ normalised,
  central moment of $X$ given $Y$ is uniformly bounded for all values
  of $Y$, i.e.
  \[\aligned
    \kappa_q &\eqdef \sup_{y} %
    { \E[\Big]{\,
        \sigma^{-q} \abs{X \!-\!  \E{X \given Y}\,}^q \given Y\!=\!y}} <
    \infty.  \endaligned\]
\end{assumption}

  Moreover, we will make repeated use of the following lemma:
  \newcommand{\ZN}{{\overline{Z}_{\!N}}}
  \begin{lemma}\label{lemma:mc-prob-bound}
    Let $\ZN$ be an average of $N$ i.i.d.\ samples of a random
    variable, $Z$, with zero mean and finite $q^\th$ moment
    \revise{for $q \geq 1$}. Then for any $z > 0$
    there exists a constant, $C_q$, depending only on $q$, such that
    \begin{eqnarray*}
      \E{\, \abs{\ZN}^q}
      &\leq& C_q\, N^{-\fpow{q}{2}}\, \E{\, \abs{Z}^q },\\[0.1in]
      \text{and } \qquad \prob{\, \abs{\ZN} \!>\! z \,}
      &\leq& \min\p*{1, C_q z^{-q} N^{-\fpow{q}{2}} \,\E{\,\abs{Z}^q}}.
    \end{eqnarray*}
  \end{lemma}
  \begin{proof}
    Denoting by $\en\{\}{Z_n}_{n=1}^N$ the $N$ samples of $Z$, the
    discrete \linebreak Burkholder-Davis-Gundy inequality %
    gives
    \begin{eqnarray*}
      \E{\, \abs{\ZN}^q}
      &\leq& C_q\, \E*{ \p*{N^{-2} \sum_{n=1}^N Z_n^2}^{\fpow{q}{2}} }\\
      &\leq& C_q\, \E*{ N^{-{q}/{2}-1} \sum_{n=1}^N \abs{Z_n}^q }\\
      &=& C_q\, N^{-\fpow{q}{2}}\, \E{\, \abs{Z}^q },
    \end{eqnarray*}
    where $C_q$ is a constant depending only on
    $q$~\cite{burkholder:1966}. The second result follows
    immediately using Markov's inequality and bounding the probability
    by 1.
  \end{proof}

  In particular, given $Y$, letting $Z=X(Y) - \E{X \given Y}$ and
  $\ZN = \inner_N(Y) - \E{X \given Y}$ and using the definition of
  $\kappa_q$, we have
  \begin{equation}\label{eq:tail-probability-bound}
    \prob[\Big]{\, \abs*{\inner_N(Y) -\E{X \given Y}} \!>\! d \,
      \given Y} \leq \min\p*{1, C_q \kappa_q \p*{\delta
        N^{\fpow{1}{2}}}^{-q}}.
  \end{equation}
  This is a generalisation of~\eqref{eq:chebyshev-prob-bound}
    when we have bounded $q$-moments.

\subsubsection{Algorithm}
Recall the choice~\eqref{eq:mc-adaptive-samples} for the adaptive
number of samples that was made in~\cite{broadie:adapt}. Another
choice can be motivated by the Central Limit Theorem. We know that a
Monte Carlo estimate of $\E{X \given Y}$ with $N$ samples, denoted by
$\inner_N$, has an error that is roughly bounded by
$C \sqrt{\var{X \given Y} / N}$, where $C$ is a confidence
constant. Then, we only need $N \geq C^2 {\sig^2} / {d^2}$ in order
to ensure that our estimate
$\heaviside{\inner_N(Y)} \approx \heaviside{\E{X \given Y}}$ is
correct. Imposing, again, a maximum number of samples of
$\Order{\tol^{-1}}$, we get
\begin{equation}\label{eq:mc-adaptive-samples-clt}
  N = \en*\lceil\rceil{\max\p*{\Order{\tol^{-1}}, C^2
      \frac{\sig^2}{{d^2}}}}.
\end{equation}
Compare this to~\eqref{eq:mc-adaptive-samples} and note the different
power of $\p*{{d}/{\sigma}}$ and the introduction of the
confidence constant $C$.

More generally, in the context of MLMC, the algorithm we would like to
use
should ensure, conditioned on $Y$, that the number of samples,
$N_\ell$, on level $\ell$, satisfies $N_\ell =
\en*\lceil\rceil{\fN_\ell}$ where
\begin{gather}\label{eq:mlmc-adaptive-samples}
  \fN_\ell = {N_0 4^{\ell}
      \max\p*{2^{-{\ell}}, \min\p*{ 1, \p*{C^{-1}
            \,N_0^{\fpow{1}{2}} \,2^{{\ell}}\, \frac{d}{\sig}
          }^{-r} }}},
  \end{gather}
for some given confidence constant $C$ and $1 < r < 2$ as we shall see
later. %
The choice~\eqref{eq:mlmc-adaptive-samples} is motivated
from~\eqref{eq:mc-adaptive-samples} and
\eqref{eq:mc-adaptive-samples-clt}, with the following changes: i) the
number of samples on level $\ell$ is now bounded below by
$N_0 2^{{\ell}}$ and bounded above by $N_0 4^{\ell}$, ii) we introduce
a generic power $r$, where, for a given value of $Y$, the number of
samples increases as $r$ decreases, see \cref{fig:Nl-vs-delta},
and iii) we introduce a confidence constant $C \geq 1$.

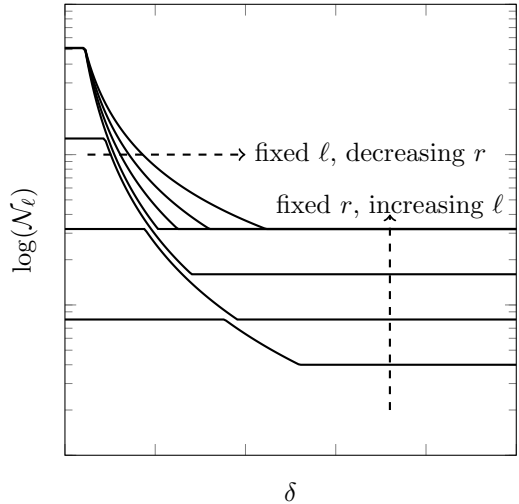
\begin{figure}[t]
  \centering
  \begin{tikzpicture}[line/.style={black, thick}]

  \begin{axis}[%
    width=\figurewidth,
    height=\figureheight,
    yticklabel style={text width=0em,align=right},
    xticklabel style={text height=0em,align=right},
    scale only axis,
    xmin=0, xmax=1,
    xlabel={$\delta$},
    ymode=log,
    ymin=1, ymax=1000,
    ylabel={$\log\p*{\fN_\ell}$},
    xticklabels={},
    yticklabels={},
    ]

    \def\Nzero{2}
    \foreach \l/\r in {4/1.2, 4/1.4, 4/1.6, 4/1.8, 3/1.8, 2/1.8, 1/1.8}
    {
      \addplot [line, domain=0:1, samples=200]
      {
        \Nzero * 4 ^ (\l) * max(2^(-\l), min(1, 1 / ( (x * sqrt(\Nzero * 4^\l)) ^ (\r)) ) )
      };
    }

    \node[right]
    at (axis cs:0.4,100) {fixed $\ell$, decreasing $r$};
    \draw [black, thick, dashed, ->] (axis cs:0.05,100) -- (axis cs:0.4,100);

    \node
    at (axis cs:0.72,45) {fixed $r$, increasing $\ell$};
    \draw [black, thick, dashed, ->] (axis cs:0.72,2) -- (axis cs:0.72,40);

  \end{axis}
\end{tikzpicture}%

  \caption{The function $\fN_\ell$ in~\eqref{eq:mlmc-adaptive-samples}
    as a function of $\delta = {d}/{\sig}$ for different
    values of $\ell$ and $r$.}\label{fig:Nl-vs-delta}
\end{figure}

  The difficulty with~\eqref{eq:mlmc-adaptive-samples} is that we do
  not know a-priori the values of $d$ and $\sigma$, given $Y$, so instead we use
  an iterative algorithm, inspired
  by~\cite[Algorithm~3]{broadie:adapt}, that doubles the number of
  samples $N_\ell$ at every iteration, then estimates $d$ and
  $\sig$. The algorithm then terminates when $N_\ell$ satisfies
  $N_\ell \geq \hfN_\ell$, where $\hfN_\ell$ is the same as
  $\fN_\ell$, but computed with the current estimates of $d$ and
  $\sigma$, denoted by $\hd$ and $\hsig$, respectively. Hence, the
  output of the algorithm, $N_\ell$ satisfies
    \begin{equation}\label{eq:nl-algorithm-output}
      \hfN_\ell \leq N_\ell < 2\, \hfN_\ell.
    \end{equation}
    The full algorithm is listed in \cref{alg:adaptive}. Note
    that the number of inner samples \revise{that} this algorithm uses to
    determine $N_\ell$ is, at most,
    $(2-2^{-\ell}) N_\ell < 2\, N_\ell$.

  \begin{algorithm}[htb]
    \KwData{$\ell, y, N_0, r$} \KwResult{$N_\ell$} \vspace*{0.1in} set
    $N_\ell \eqdef {N_0 2^{\ell}}$\; set done $\eqdef$ false\;
    \Repeat{done}{ \tcc{If we continue the adaptive algorithm, at
        least $2 N_\ell$ inner samples will be used ($N_\ell$ samples
        in this step and then another $N_\ell$ samples when computing
        the Monte Carlo estimate). Hence, terminate if $2N_\ell$
        is greater than the maximum number of inner samples.}
      \eIf{$2 N_\ell \geq N_0 4^{\ell}$} { set
        $N_\ell \eqdef N_0 4^{\ell}$\; set done $\eqdef$ true \; } {
        generate $N_\ell$ new, and independent, inner samples\;
        compute new estimates $\hd$ and $\hsig^2$ given $Y=y$\;
        \eIf{$\displaystyle N_\ell \geq N_0 4^{\ell}\, {\p*{C^{-1}
              N_0^{\fpow{1}{2}}2^{{\ell}}\frac{\ \hd \ }{\ \hsig\
              }}^{-r}}$}{ set done $\eqdef$ true\; } { set
          $N_\ell \eqdef 2 N_\ell $ } } } \vspace*{0.1in}
    \Return{} $N_\ell$\;
    \caption{Adaptive algorithm to determine $N_\ell$.}\label{alg:adaptive}
  \end{algorithm}

\subsubsection{Numerical analysis}
In this section we will show that, under
\cref{ass:bounded-density,ass:bounded-moments}, the adaptive algorithm has a
random output, $N_\ell$, that satisfies $\E{N_\ell} = \Order{2^{\ell}}$.
Moreover, we will show that, with this random choice of number of inner samples,
we have $\var*{\heaviside{\inner_{N_\ell}(Y)} - \heaviside{\E{X \given Y}}} =
\Order{2^{-\ell}}$. From here, using standard MLMC complexity analysis, the
complexity of our method is $\Order{\tol^{-2} \logtol2}$. We start in
\cref{thm:perfect-adaptivity} by proving the results when we have perfect
knowledge of $d$ and $\sigma^2$, then in \cref{thm:adaptive-hd} we consider the
case when $\hd$ and $\hsig^2$ are Monte Carlo estimates of $d$ and $\sig^2$,
given $Y$.
\revise{In what follows, we will denote a ``normalised $\delta$'' by
\begin{equation}\label{eq:z-def}
  \ndelta \eqdef {C^{-1} N_0^{\fpow{1}{2}} 2^{\ell} \delta},
\end{equation}}

\begin{lemma}[Perfect adaptive sampling]\label{thm:perfect-adaptivity}
  Under \cref{ass:bounded-density,ass:bounded-moments}, for the
  estimator~\eqref{eq:mlmc-inner-level} and the number of samples $N_\ell$
  obtained using \cref{alg:adaptive} for $1 < r < 2 - {2}/{q}$,
  $\hd=d={\abs{\E{X \given Y}}}$, and
  $\hsig^2=\sigma^2={\var{X \given Y}}$ we have
  \[
    \E{N_\ell} = \Order{2^{\ell}} \quad\text{and}\quad
    \var*{\heaviside{\inner_{N_\ell}(Y)} - \heaviside{\E{X \given Y}}}
    = \Order{2^{-\ell}}.\]
\end{lemma}
\noindent \emph{Comment.} Even though the setting of this lemma is
idealistic, in that it assumes perfect knowledge of $d$ and
$\sigma^2$, its proof still illustrates important points. The first
point is the usage of~\eqref{eq:tail-probability-bound} to bound the
tail probability of a Monte Carlo average. The second point is that
the bounds on the value of $r \in (1,2-2/q)$ are needed even with such
perfect knowledge.

  As stated previously, as $r$ increases, the required number of inner
  samples, for the same value of $y$, decreases. This means that,
  subject to the condition $r < 2 - {2}/{q}$ or when
  $q \to \infty$, we want to take $r$ as close as possible to 2 to
  reduce the total number of needed inner samples while still
  maintaining the same variance convergence rate. On the other hand,
  if, for example, we only have bounded fourth moments, i.e., $q=4$,
  then we must have $r < 3/2$ to have the same convergence rate of the
  variance.
\begin{proof}
  In this lemma, we are supposing that we have perfect knowledge of
  $d$ and $\sigma$, and hence $\hfN = \fN$. Recall that the output of
  the adaptive algorithm satisfies~\eqref{eq:nl-algorithm-output}.
  Then, given $N_\ell < 2\, \fN_\ell$ we have,
\begin{align}\label{eq:work-bound}  \E{N_\ell} &\leq 2 N_0 4^{\ell}\int_0^\infty
               \max\p*{2^{-\ell},
               \min\p*{1,\p*{C^{-1} N_0^{\fpow{1}{2}} 2^{\ell}\, \delta}^{-r}}}
               \rho(\delta) \D \delta \\
  \notag  &\leq 2 N_0 2^{\ell} + 2 N_0 4^{\ell}\int_0^\infty
               \min\p*{1,\p*{C^{-1} N_0^{\fpow{1}{2}} 2^{\ell}\, \delta}^{-r}}
               \rho(\delta) \D \delta \\
  \notag &\leq
           2 N_0 2^{\ell} + 2 N_0 4^{\ell}\rho_0 \int_0^\infty
           {\min\p*{1,\p*{C^{-1} N_0^{\fpow{1}{2}} 2^{\ell}\, \delta}^{-r}}}
           \D \delta \\
  \notag &\hskip 2cm + 2 N_0 4^{\ell} {\min\p*{1,\p*{C^{-1} N_0^{\fpow{1}{2}} 2^{\ell}\,
               \delta_0}^{-r}}} \\
  \notag &\leq
           2 N_0 2^{\ell} +  2 C N_0^{\fpow{1}{2}} \rho_{0}
           \frac{r}{r-1} 2^{\ell} + 2 C^{r} N_0^{(2-r)/2}\,\delta_0^{-r} 2^{(2-r)\ell},
\end{align}
where we first used~\eqref{eq:int-delta-identity} and
then~\eqref{eq:int-min-identity}. Since $2-r < 1$, we have
$\E{N_\ell} = \Order{2^{\ell}}$.

On the other
hand, by~\eqref{eq:tail-probability-bound},
\[\aligned
  \var*{\heaviside{\inner_{N_\ell}(Y)} \!-\! \heaviside{\E{X \given Y}}} &\leq
  \E*{\prob*{\abs*{\, \inner_{N_\ell}(Y) \!-\! \E{X \given Y} \,} \!>\!
      d \given Y}} \\ &\leq \E*{\min\p*{1, \ C_q\, \kappa_q
      \p*{\frac{d}{\sigma} N_\ell^{{1}/{2}}}^{-q}}}.
  \endaligned
\]
\revise{Using the definition of $\ndelta$ in \eqref{eq:z-def} and
  since $N_\ell \geq \fN_\ell$}, we have
  \begin{equation}\label{eq:var-bound}
  \begin{aligned}
     \var{\heaviside{\inner_{N_\ell}(Y)} \!-\!
    \heaviside{\E{X \given Y}}}
    &\leq \E*{\min\p*{1, \ C_q\, \kappa_q
      C^{-q} {\ndelta^{-q}}\p*{\frac{ N_\ell}{N_0 4^{\ell}}}^{-\fpow{q}{2}}}}\\
    \notag & \hskip -2.5cm \leq  \E*{\min\p*{1, \ C_q\kappa_q C^{-q}\,  {\ndelta^{-q}}
             \min\p*{2^{\fpow{\ell q}{2}}, \max\p{1, {\ndelta^{\fpow{rq}{2}}}}}}}\\
         \notag & \hskip -2.5cm \leq \rho_0 \int_0^{\infty}\!\!
         \min\p*{1, \ C_q\kappa_q C^{-q}\, {\ndelta^{-q}}
           \max\p*{1, \ndelta^{\fpow{rq}{2}}}} \D \delta \\
         \notag & \hskip -1.5cm + \min\p*{1, \ C_q\, \kappa_q
           2^{-\fpow{\ell q}{2}}
           N_0^{-\fpow{q}{2}}\delta_0^{-q}} \\
         \notag &\hskip -2.5cm \leq \rho_{0} C N_0^{-\fpow{1}{2}}
         2^{-\ell} \int_0^{\infty}\!\! \min\p*{1, \ C_q \kappa_q
           C^{-q} \, \ndelta^{-q}
           \p*{1 + \ndelta^{\fpow{rq}{2}}}} \D \ndelta \\
           \notag & \hskip -1.5cm + \ C_q\, \kappa_q 2^{-\fpow{\ell
               q}{2}} N_0^{-\fpow{q}{2}} \delta_0^{-q},
  \end{aligned}
\end{equation}
  where we first used~\eqref{eq:int-delta-identity} and then changed
  the variable of integration from $\delta$ to $\ndelta$. The integral in
  the first term is a constant independent of $\ell$ and can be
  computed using~\eqref{eq:int-min-identity}
  since $q - rq/2> 1$. Hence, %
  since $q > 2$ the variance is $\Order{2^{-\ell}}$.
\end{proof}

\def\pq{{q}}  %
\def\vq{{p}}  %

{
  \begin{theorem}[Adaptive sampling with $\hd$ and
    $\hsig$-estimators]\label{thm:adaptive-hd}
    Let \cref{ass:bounded-density,ass:bounded-moments} hold and consider the
    estimator~\eqref{eq:mlmc-inner-level} and the number of samples $N_\ell$
    obtained using \cref{alg:adaptive} for{
      \begin{equation}\label{eq:r_adaptive_bound}
    1 < r < 2 - \frac{\sqrt{4q+1}-1}{q}.
  \end{equation}}
Assume further that in every iteration of the algorithm given $Y$, the
current number of samples, $N_\ell$, and
$\en\{\}{x^{(m)}(Y)}_{n=1}^{N_\ell}$ being i.i.d.\ samples of $X$
given $Y$, we estimate $d$ and $\sigma^2$ by
\begin{align}
  \label{eq:hd-estimator}\hd &= \abs*{\inner_{N_\ell}\p{Y}} = \abs*{\frac{1}{N_\ell}\sum_{n=1}^{N_\ell} x^{(n)}(Y)}\\
  \label{eq:hsig-estimator} \text{and} \hskip 1cm \hsig^2 &= \frac{1}{{N_\ell}} \sum_{n=1}^{{N_\ell}}
                                                            \p*{x^{(n)}(Y) - \inner_{N_\ell}\p{Y}}^2,
\end{align}
respectively. Then we have
\[
  \E{N_\ell} = \Order{2^{\ell}} \quad \text{and} \quad
  \var*{\heaviside{\inner_{N_\ell}(Y)} - \heaviside{\E{X \given Y}}} =
  \Order{2^{-\ell}}.
\]
\end{theorem}
\revise{Before proving this theorem, we will bound the probability of
  incurring a given error when estimating the variance, $\sigma^2$,
  with~\eqref{eq:hsig-estimator}.}

\begin{corollary}\label{thm:variance-err-bound-det}
  \revise{Let \cref{ass:bounded-moments} hold.
  For a fixed $\ell'=\ell \ldots 2\ell$ and a given $Y$, denote by
  $\hsig_{\ell'}$ the estimate of $\sig$ computed using
  $N_0 2^{\ell'}$ samples of $X$ given $Y$ by
  \begin{equation}\label{eq:sig-estimator}
    \hsig_{\ell'}^2 = \frac{1}{N_0 2^{\ell'}} \sum_{n=1}^{N_0
      2^{\ell'}} \p*{x^{(n)}(Y) - \inner_{N_0
        2^{\ell'}}\p{Y}}^2,
  \end{equation}}
  Then, for any constant $c_1 > 0$, there exists a constant $c_2$,
  depending only on $N_0, \kappa_q, q$ and $c_1$, such that
  \[
    \prob[\Big]{\abs*{\hsig_{\ell'}^2 - \sig^2} > c_1 \sig^2 \given
      Y} \leq c_2 2^{-q\ell'/4},
  \]
\end{corollary}
\begin{proof}
  We can write
    \[
    \hsig^2_{\ell'} = \frac{1}{N_0 2^{\ell'}}\sum_{n=1}^{N_0 2^{\ell'}} \p*{x^{(n)}(Y) -
      \E{X \given Y}}^2 - \p*{\inner_{N_0 2^{\ell'}}(Y)-\E{X \given Y}}^2.
  \]
  Next, using \cref{lemma:mc-prob-bound} and the fact that
  $\E{\p*{X - \E{X\given Y}}^2 \given Y} = \sigma^2$, yields
  \[
      \begin{aligned}
        &\prob[\Big]{\abs*{\hsig^2_{\ell'} - \sig^2 } > c_1 \sig^2
          \given Y} %
        \\
        &\leq \prob*{\abs*{ \frac{1}{N_0 2^{\ell'}}\sum_{n=1}^{N_0
              2^{\ell'}} \p*{x^{(n)}(Y) - \E{X \given Y}}^2 - \sig^2}
          > \frac{1}{2}
          c_1 \sig^2 \; \given \; Y} \\
        & \hskip 1cm + \prob*{\abs*{\inner_{N_0 2^{\ell'}}(Y) - \E{X
              \given Y} } >
          \p*{\frac{c_1}{2}}^{1/2} \sig \given Y} \\
        &\leq 2^{q/2}\; C_{q/2}\; c_1^{-q/2}\;\kappa_{q}\; \p*{N_0
          2^{\ell'}}^{-q/4} \\
        & \hskip 1cm+ 2^{q/2}\; C_{q} \;
        c_1^{-q/2}\;\kappa_{q}\; \p*{N_0 2^{\ell'}}^{-q/2}\\
        &\leq c_2 2^{-q\ell'/4}.
      \end{aligned}
    \]
  \end{proof}
\revise{
  \begin{corollary}\label{thm:variance-err-bound}
    Let \cref{ass:bounded-moments} hold and assume that at
    every iteration in the loop inside \cref{alg:adaptive},
    $\sig$ is estimated by $\hsig$ as
    in~\eqref{eq:hsig-estimator}. Then, there exists a constant $c_3$
    such that
    \[
      \prob[\Big]{\abs{\hsig^2 - \sig^2} > c_1 \sig^2 \given Y}
      \leq c_3 2^{-q\ell/4},
    \]
    In particular, this is true for the \emph{final} estimate of
    $\sigma$ computed in \cref{alg:adaptive} before returning
    $N_\ell$.
\end{corollary}
\begin{proof} Using \cref{thm:variance-err-bound-det},
  \[\aligned
    \prob*{\abs*{\hsig - \sig^2} > c_1 \sig^2 \given Y } &\leq
    \sum_{\ell'=\ell}^{2\ell} \prob[\Big]{\abs*{\hsig_{\ell'} -
        \sig^2} >
      c_1 \sig^2 \given Y }\\
    &\leq c_2 \sum_{\ell'=\ell}^{2\ell} 2^{-q\ell'/4} \leq
    \frac{c_2}{1-2^{-q/4}} 2^{-q \ell /4}.  \endaligned
    \]
  \end{proof}
}
  We are now ready to prove \cref{thm:adaptive-hd}.

\begin{proof}
  First, for a given value of $Y$ and a corresponding
  $\delta = {d}/{\sigma}$, assume that, with perfect knowledge of
  $\delta$, the adaptive algorithm would return $N_0 2^{\ell^*}$ as
  the number of inner samples, for $\ell \leq \ell^* \leq
  2\ell$. Given the stopping condition, we have
    \begin{equation}\label{eq:ellstar-cond}
  \frac{2^{(\ell^*-1)}}{4^{\ell}} \ <\ \p*{C^{-1} N_0^{\fpow{1}{2}}
    2^{\ell} \delta }^{\!\!-r} \leq\ \frac{2^{\ell^*}}{4^{\ell}},
\end{equation}
whenever $\ell < \ell^* \leq 2\ell$ for the left inequality
  and whenever $\ell \leq \ell^* < 2\ell$ for the right inequality.
  \revise{Moreover, suppose that, for a given $Y$,
    \cref{alg:adaptive} returns
    $N_\ell = N_0 2^{\widehat \ell}$ for
    $\ell \leq \widehat \ell \leq 2\ell$ and let $\hd$ and $\hsig^2$
    be the \emph{final} estimates of $d$ and $\sig$, respectively,
    that were computed in \cref{alg:adaptive}.}
The
proof proceeds by analysing the probability that $\widehat \ell$
differs significantly from $\ell^*$ to bound the variance and
work. In what follows, the constants in the $\Order{\cdot}$
  notation depend on $\kappa_q, q, N_0, C$ and $r$ only.

  \subparagraph{Bounding the work} To bound the work, we consider the
  case in which the adaptive algorithm terminates on a level $\ell'$
  such that $\ell' \geq \ell^* + 3$ for
  $\ell \leq \ell^* \leq 2\ell-3$, i.e., the case of \revise{returning
    too many inner samples; where ``too many'' here means that the
    adaptive algorithm returned a factor of at least $2^3$ more samples than it
    would have returned had we used $d$ and $\sig$ instead of their
    approximations, $\hd$ and $\hsig$, respectively. The choice
    $\ell^*+3$ is somewhat arbitrary; $\ell^*+2$ is the minimum to
    ensure the positivity of certain terms in the proof, but this
    particular choice simplifies the subsequent algebra.
  } In any cases, $\ell' \geq \ell^* + 3$
  implies that the termination condition of the adaptive algorithm was
  \emph{not} satisfied in each $\p{\ell''-\ell}^\th$ iteration of the
  algorithm where $\ell^*+ 2 \leq \ell'' < \ell'$.  This, with the
  right inequality in~\eqref{eq:ellstar-cond}, yield
\[
\p*{{C^{-1} N_0^{1/2}2^{\ell} \hdelta_{\ell''}}}^{\!\!-r}
\!\geq\, \frac{2^{\ell''}}{4^\ell}
\,\geq\, 4\, \frac{2^{\ell^*}}{4^\ell}
\,\geq \, 4 \p*{C^{-1}  N_0^{1/2}{2^{\ell} \delta }}^{\!\!-r}
~~~\Longrightarrow ~~~
\hdelta_{\ell''} < 4^{-1/r} \delta,
\]
where $\hdelta_{\ell''} = \hd_{\ell''}/\hsig_{\ell''}$ and
$\hd_{\ell''}$ and $\hsig_{\ell''}^2$ denote the Monte Carlo estimates
of $d$ and $\sig^2$, respectively, using $N_0 2^{\ell''}$ inner
samples. Then, since the inner samples used in the iterations of
\cref{alg:adaptive} are mutually independent, we have
\[\aligned
  \prob{\widehat \ell = \ell' \given Y}\, &\leq\, \prod_{\ell'' =
    \ell^*+2}^{\ell'-1}
  \prob*{\frac{d}{\hd_{\ell''}}\cdot\frac{\hsig_{\ell''}}{\sig} >
    4^{1/r} \given Y } \,\\
  &\leq \prod_{\ell'' = \ell^*+2}^{\ell'-1} \p*{\prob*{d > 2^{1/r}
      \hd_{\ell''} \given Y} + \prob*{\hsig_{\ell''}^2 >
      4^{1/r}\sig^2\given Y}}, \endaligned
\]
for any $\ell'$ such that $\ell' \geq \ell^{*} +3$.  Using the right
inequality in~\eqref{eq:ellstar-cond}
and~\eqref{eq:tail-probability-bound}, yields
\[
  \begin{aligned}
    & \prob*{d > 2^{1/r} \hd_{\ell''} \given Y} =\,
    \prob*{d -\hd_{\ell''} > \p*{1-2^{-1/r}} d\given Y}\\
    & \hskip 0.3cm \leq\, \prob*{\,\abs*{\inner_{N_0 2^{\ell''}}\p{Y} - \E{X \given
          Y}} > \p*{1-2^{-1/r}} C \sig N_0^{-\fpow{1}{2}} \, 2^{-\ell}
      2^{(2\ell-\ell^*)/r}\given Y} \\
    &\hskip 0.3cm  \leq \p*{1-2^{-1/r}}^{-q} C_\pq \kappa_\pq C^{-\pq}
    \,2^{-\pq(\p{2\ell-\ell^*}/{r}-\p{2\ell-\ell''}/{2})}\\
    &\hskip 0.3cm \leq \Order{2^{-q(\ell'' - \ell^*)/2}}.
  \end{aligned}
\]
Moreover, using \cref{thm:variance-err-bound}, we
  have, for any $Y$,
\[
  \aligned \prob*{\hsig_{\ell''}^2 > 4^{1/r}\sig^2 \given Y}
  &\leq \prob*{\abs{\hsig_{\ell''}^2 - \sig^2} > \p{4^{1/r}-1}\sig^2\given Y}\\
  &\leq \Order{ 2^{-q(\ell'' - \ell^*)/4}}.  \endaligned
\]
Hence for any $\ell'$ such that
  $\ell^{*} +3\leq \ell' \leq 2\ell$ and
  $\ell \leq \ell^* \leq 2\ell -3$, we have
    \begin{equation}\label{eq:too-many-samples-prob}
  \prob{\widehat \ell = \ell' \given Y}\, \leq
  \prod_{\ell''=\ell^*+2}^{\ell'-1} \Order{ 2^{-q(\ell''-\ell^*)/4}}
  \leq \Order {2^{-q\p{\ell'-\ell^*}^2 / 8}}.
\end{equation}
In general, this bounds
  $\prob{\widehat \ell = \ell' \given Y}$ for all
  $\ell \leq \ell^* \leq 2\ell$ since
  $\prob{\widehat \ell = \ell' \given Y} = 0$ for
  $\ell^*+3 \leq \ell' \leq 2\ell$ whenever
  $2\ell-3 < \ell^* \leq 2\ell$.  Next, conditional on $Y$ we compute
the expected number of samples, $N_\ell$, as follows
\[\aligned
  \E*{\, \frac{N_\ell}{N_0 4^\ell} \given Y } &\ =\
  \sum_{\ell'=\ell}^{2 \ell} \prob{N_\ell = N_0 2^{\ell'} \given Y}
  2^{\ell'-2\ell} \ \\
  &\leq\ \revise{2^{\ell^*+3-2\ell}} + \sum_{\ell'=\ell^*+3}^{2 \ell}
  \prob{\widehat \ell = \ell' \given Y} 2^{\ell'-2
    \ell}.\\
  &\leq \revise{2^{\ell^*+3-2\ell}} + \Order*{2^{\ell^*-2\ell}}
  \sum_{\ell'=\ell^*+3}^{2\ell} {2^{-q\p{\ell'-\ell^*}^2/8 + \ell'
      -\ell^*}}, \endaligned\] where we
substituted~\eqref{eq:too-many-samples-prob} in the last step. Since
the sum in the previous equation is bounded for any $q$, we have,
using the bound on $\ell^*$
in~\eqref{eq:ellstar-cond}, the fact that
$\ell \leq \ell^*  \leq 2\ell$ and the definition of $\ndelta$
in~\eqref{eq:z-def}, that
\[\aligned
  \E*{\, \frac{N_\ell}{N_0 4^\ell} \given Y }& \ =\ \Order*{2^{\ell^*-2
      \ell}} \\
   &\ =\ \Order{\min\p*{1, \max\p*{2^{-\ell}, {\ndelta^{-r}}}}},
  \endaligned
\]
for all values of $\ndelta$, and a similar calculation
to~\eqref{eq:work-bound} gives the desired bound on the overall
expected number of inner samples.

\subparagraph{Bounding the variance} When bounding the variance, we
want to control the probability of $\widehat \ell$ being significantly
smaller than $\ell^*$, i.e., the \revise{probability of returning a
  significantly smaller number of inner samples than we would have
  returned had we used $d$ and $\sig$ instead of their approximations,
  $\hd$ and $\hsig$, respectively}.  In particular,
$\widehat \ell < \ell^*$ implies that $\hd$ over-estimates $d$ and/or
$\hsig$ under-estimates $\sig$. We first deal with the latter case by
denoting
$G_\ell \eqdef \heaviside{\inner_{N_\ell}(Y)} - \heaviside{\E{X
    \given Y}}$ and writing the variance as
  \begin{equation}\label{eq:var-decompose}
  \aligned \var*{G_\ell} \leq \E*{\E{G_\ell^2 \given Y}} &=
  \E*{\;\E{G_\ell^2 \given Y,\; b^2
      \sig^2 > \hsig^2} \;\prob{b^2 \sig^2 > \hsig^2 \given Y}\;} \\
  &\hskip 1cm + \E*{\;\E{G_\ell^2 \;\ind{b^2  \sig^2 \leq \hsig^2} \given
      Y}\;
  }.
  \endaligned
\end{equation}
for some constant $0<b<1$, independent of $\ell, \ell^*$ and
$\widehat \ell$.  The first term in~\eqref{eq:var-decompose} deals with the
case when $\hsig$ under-estimates $\sigma$ while the second term deals
with the opposite case. Using \cref{thm:variance-err-bound}, we have, for any $Y$,
\[
  \begin{aligned}
    \prob*{b^2 \sig^2 > \hsig^2 \given Y}
    \leq \prob*{\abs{\sig^2- \hsig^2} > \p*{1-b^2}\sig^2  \given Y}
    \leq \Order{2^{-\pq \ell/4}}.
  \end{aligned}
\]
On the other hand, since the inner samples used to compute $G_\ell$
are independent from those used to compute $N_\ell$, we can bound, for
any $Y$,
\[\aligned
  {\E{G_\ell^2 \given Y,\; b \sig^2 > \hsig^2}}
  &= \E{\E{G_\ell^2 \given N_\ell} \given Y, \; b \sig^2 > \hsig^2} \\
  &= \E{\prob{\abs{\inner_{N_\ell}(Y) - \E{X \given Y}}
      \geq d \given N_\ell}\given Y, \; b \sig^2 > \hsig^2} \\
  &\leq \E{\min\p*{1, \delta^{-2} N_\ell^{-1}}\given Y, \; b \sig^2 > \hsig^2} \\
  & \leq {\min\p*{1, \delta^{-2} N_0^{-1} 2^{-\ell}}}.
  \endaligned
\]
Here, we used~\eqref{eq:chebyshev-prob-bound}, then the fact that
$N_\ell \geq N_0 2^{\ell}$. Taking expectation with respect to $Y$ and
using~\eqref{eq:int-delta-identity} and~\eqref{eq:int-min-identity} we
have that
\[
  \E*{\E{G_\ell^2 \given Y,\; b \sig^2 > \hsig^2}} \leq 2 \rho_0 N_0^{-1/2} 2^{-\ell/2} + \delta_0^{-2} N_0^{-1}
  2^{-\ell}.
\]
Hence, we have that the first term
in~\eqref{eq:var-decompose} is
$\Order{2^{-\ell/2 - q\ell/4}} = \Order{2^{-\ell}}$, since $q \geq 2$.

Now, we turn our attention to the second term
in~\eqref{eq:var-decompose} where the estimator $\hsig$ does not
significantly under-estimate $\sig$, i.e., given
$b^2 \sigma^2 \leq \hsig^2$. First, for some
$\vq \in ({2}/\p{2-r}, q)$, we bound by~\eqref{eq:tail-probability-bound}
\[
  \begin{aligned}
    {\;\E{G_\ell^2 \;\ind{b^2 \sig^2 \leq \hsig^2} \given Y}} &=
    \E*{\;{\E{G_\ell^2 \given Y, \hd, \hsig^2} \;
        \ind{b^2 \sig^2 \leq \hsig^2} }\given Y}\\
    &\leq \E*{\;{\prob*{\abs*{\, \inner_{N_\ell}(Y) \!-\! \E{X \given Y}
            \,} \!>\!  d \given Y, \hd, \hsig^2} \;
        \ind{b^2 \sig^2 \leq \hsig^2} }\given Y}\\
    &\leq \E*{\;{\min\p*{1, \ C_\vq\, \kappa_\vq \p*{\delta
            N_\ell^{{1}/{2}}}^{-\vq}} \;
        \ind{b^2 \sig^2 \leq \hsig^2} }\given Y}\\
    &\leq \E*{\;{ \min\p*{1, \ C_\vq\, \kappa_\vq \p*{\delta
            N_\ell^{{1}/{2}}}^{-\vq} \ind{b^2
            \sig^2 \leq \hsig^2}} \;}\given Y} \\
    &\leq {{\min\p*{1, \ C_\vq\, \kappa_\vq
          C^{-\vq}\ndelta^{-\vq}\E*{\p*{\frac{N_\ell}{N_0
                4^\ell}}^{{-\vq}/{2}} \ind{b^2 \sig^2 \leq \hsig^2}
            \;\given Y}} }}.
    \end{aligned}
  \] Here we used that fact that the inner samples used to compute
  $G_\ell$ are independent from those used to compute $\hd$ and
  $\hsig$.  We will proceed by bounding the probability of $N_\ell$
  being too small, then using this probability to bound the
  conditional expectations
  $\E*{\p*{N_\ell / \p*{N_0 4^{\ell}}}^{-\fpow{\vq}{2}} \ind{b^2
      \sig^2 \leq \hsig^2} \given Y}$. Finally, we conclude by
  following the same steps as in the proof of
  \cref{thm:perfect-adaptivity}.  To start, the question we will
  address is: what is the probability that the adaptive algorithm
  terminates on a level $\ell'$ given that $\ell' \leq \ell^* - 3$ and
  $b^2 \sigma^2 \leq \hsig^2$?  \revise{I.e., the probability of
    returning too few inner samples; where ``too few'' here means that
    the adaptive algorithm returned a fraction of at most $1/2^3$
    samples of what it would have returned had we used $d$ and $\sig$
    instead of their approximations, $\hd$ and $\hsig$,
    respectively. The choice $\ell^*-3$ is somewhat arbitrary;
    $\ell^*-2$ is the maximum to ensure the positivity of certain
    terms in the proof but this particular choice simplifies the
    subsequent algebra.
  }In any case, $\ell' \leq \ell^* - 3$, the termination condition in
  \cref{alg:adaptive} and the left inequality
  in~\eqref{eq:ellstar-cond} imply that
\[
\p*{C^{-1}N_0^{1/2} 2^{\ell} \hdelta }^{\!\!-r}
\!\leq\, \frac{2^{\ell'}}{4^\ell}
\,\leq\, \frac{1}{4} \cdot \, \frac{2^{(\ell^*-1)}}{4^\ell}
\,<\, \frac{1}{4}  \p*{C^{-1}N_0^{1/2} 2^{\ell} \delta }^{\!\!-r}
~~~\Longrightarrow ~~~
{\hdelta} > 4^{1/r} {\delta},
\]
for $\ell \leq \ell' \leq \ell^*-3$ and
  $\ell+3 \leq \ell^* \leq 2\ell$.  Hence, choosing $b$ such that
$4^{-1/r} < b < 1$, we have
\def\curcond{Y}
\def\givcurcond{\given\curcond}
\def\curand{,\;b^2 \sigma^2 \leq\hsig^2}
\[
  \begin{aligned}
    &\prob{\widehat \ell = \ell' \curand \givcurcond}\, \leq\,
    \prob*{\frac{\hd}{d}\cdot\frac{\sig}{\hsig} >
      4^{1/r} \curand\givcurcond} \,\\
    &\hskip 0.8cm \leq \prob*{\hd > 4^{1/r} b d \curand\givcurcond} \\
    &\hskip 0.8cm \leq \prob*{ \hd - d > \p*{1-4^{-1/r}b^{-1}}\hd \curand\givcurcond}\\
    &\hskip 0.8cm \leq\, \prob*{\, \abs{\hd - d} > \p*{b-4^{-1/r}} C
      \sig N_0^{-{1}/{2}} \, 2^{-\ell}
      2^{(2\ell-\ell')/r}\curand\givcurcond} \\
    &\hskip 0.8cm \leq\, \prob*{\, \abs*{ {\inner_{N_0 2^{\ell'}}(Y)} -
        \E{X \given Y} } > \p*{b-4^{-1/r}} C \sig N_0^{-{1}/{2}} \,
      2^{-\ell}
      2^{(2\ell-\ell')/r}\givcurcond} \\
    &\hskip 0.8cm \leq \p*{b-4^{-1/r}}^{-\pq} C_\pq \kappa_\pq C^{-\pq}
    2^{-\pq(2\ell-\ell') \p*{{1}/{r} - {1}/{2}}}.
  \end{aligned}
  \]
  Moreover, $\prob{\widehat \ell = \ell' \givcurcond} = 0$
    for $\ell \leq \ell' \leq \ell^*-3$ whenever
    $\ell \leq \ell^* < \ell+3$. Hence
    \begin{equation}\label{eq:too-few-samples-prob}
      \prob{\widehat \ell = \ell'
      \curand\givcurcond}\, = \, \Order{
      2^{-\pq(2\ell-\ell') \p*{{1}/{r} - {1}/{2}}}},
  \end{equation}
  for $\ell \leq \ell' \leq \ell^*-3$ and
    $\ell \leq \ell^* \leq 2\ell$.
  Now, we have that
\[
  \begin{aligned}
    & \E*{\p*{\frac{N_\ell}{N_0 4^\ell}}^{\!\!\!-\vq/2} \ind{b^2 \sig^2
        \leq \hsig^2} \givcurcond } \\
    &\hskip 2cm =\ \sum_{\ell'=\ell}^{2\ell} \prob{N_\ell = N_0
      2^{\ell'}\curand \givcurcond} \;
    2^{\vq(2\ell-\ell')/2} \\
    & \hskip 2cm \leq\ 2^{\vq(2\ell-\ell^*+2)/2} +
    \sum_{\ell'=\ell}^{\ell^*-3} \prob{\widehat \ell = \ell'\curand
      \givcurcond}\; 2^{\vq(2\ell-\ell')/2}.
  \end{aligned}
\]
Substituting~\eqref{eq:too-few-samples-prob} and denoting
$u \eqdef {{\pq} \p{r-2}/{r} + \vq}$ yields
\[\aligned
  \E*{ \p*{\frac{N_\ell}{N_0 4^\ell}}^{\!\!\!-\vq/2}\ind{b^2
        \sig^2 \leq \hsig^2} \givcurcond } \ &\leq\
  2^{\vq(2\ell-\ell^*+2)/2} + \sum_{\ell'=\ell}^{\ell^*-3}
  \Order{2^{u(2\ell-\ell')/2}} \\
  &\leq\ 2^{\vq(2\ell-\ell^*+2)/2} + \Order*{\frac{2^{u
        (2\ell-\ell^*+2)/2} - 2^{u\ell/2}}{2^{-u/2}-1}}. \endaligned
\]
Note that $\p{2-r}^2\pq = 2r$ has roots $r = 2 - (\pm\sqrt{4q+1}-1)/q$
and hence from~\eqref{eq:r_adaptive_bound} we have that
$q > 2r / \p{2-r}^2$. Therefore, we can choose $p$ to satisfy
${2}/\p{2-r} < \vq < {\pq} \p{2 - r}/{r}$ in which case we
have $u< 0 < \vq$ and the previous conditional expectation is
$\Order{2^{\vq(2\ell-\ell^*)/2}}$ whenever
$\ell \leq \ell^*  \leq  2\ell$. Using the bound on $\ell^*$
in~\eqref{eq:ellstar-cond}, the fact that
$\ell \leq \ell^* \leq 2\ell$ and the definition of $\ndelta$
in~\eqref{eq:z-def}, yields
\[\aligned \E*{\, \p*{\frac{N_\ell}{N_0 4^\ell}}^{\!\!\!-\vq/2}
    \ind{b^2 \sig^2
        \leq \hsig^2} \givcurcond }
  &= \Order{2^{\vq(2\ell-\ell^*)/2}} \\
  &= \Order*{\max\p*{1, \min\p*{2^{\ell \vq/2}, {\ndelta^{r\vq/2}}}}}.
  \endaligned \]
for all values of $\ndelta$.
The condition $p > 2 / (2-r)$ implies that $r < 2 - 2/p$ and hence a
similar calculation to~\eqref{eq:var-bound} yields that the second
term in~\eqref{eq:var-decompose} is again $\Order{2^{-\ell}}$. Hence
$\var{G_\ell} = \Order{2^{-\ell}}$.

\end{proof}
}
\section{A Model Problem}\label{s:model-example}
In this section, we will look at a simple example that mimics many of
the challenges of computing the probability of a large loss from a
financial portfolio. In fact, the underlying model problem can be seen
as the loss from a single option with a stock following a Brownian
Motion and a final payoff function, $f(x) = -x^2$ evaluated at
maturity, $T=1$. This is a model for a delta-hedged portfolio with
negative Gamma such that a large loss is incurred with very low
probability under extreme circumstances.

  \paragraph{Problem setup}
  Defining, for $\tau \ll 1$,
  \[
    \begin{aligned}
      P(y, z) &\eqdef -\p*{\tau^{1/2} \;y + \p{1\!-\!\tau}^{1/2}\; z}^2\\
      &= -\tau y^2 - 2 \tau^{1/2} \p{1\!-\!\tau}^{1/2} y z -
      \p{1\!-\!\tau} z^2,
    \end{aligned}
  \]
  we will compute the following
    \begin{equation}\label{eq:eta-model}
    \eta = \E*{\heaviside{\vphantom{\Big()}
        \E*{P(\widetilde Y,Z)} - \E*{P(Y,Z) \given Y} - L_{\eta}}},
  \end{equation}
  where $\widetilde Y, Y$ and $Z$ are independent standard normal
  random variables and $L_\eta$ is a constant. In financial
  applications, $\eta$ corresponds to the probability the portfolio
  loss exceeding a given loss level, $L_\eta$, over a short risk
  horizon, $\tau$, with the portfolio loss defined as the difference
  between current and future risk-neutral portfolio expectations.
    Alternatively, we will later specify $\eta \ll 1$ and determine
    the corresponding loss level, $L_\eta$ using the
    relation~\eqref{eq:eta-model}.
    Note that only the second inner expectation
    in~\eqref{eq:eta-model} is conditioned on samples of the outer
    random variable, $Y$. When approximating the inner expectations
    for a given $Y$, we could use the exact value of
    $\E{P(\widetilde Y,Z)}=-1$ and set $X \eqdef -1-P(Y,Z) - L_\eta$.
    However, the variance is
    $\var{X\given Y} = 4\tau(1-\tau) Y^2 + 2 \p{1-\tau}^2 =
    \Order{1}$.  We could also use independent samples of $Z$ to
    compute both $\E*{P(\widetilde Y,Z)}$ and $\E*{P(Y,Z) \given Y}$,
    setting $X \eqdef P(\widetilde Y,\widetilde Z)-P(Y,Z)-L_\eta$ for
    independent standard normal variables $Z$, $\widetilde Z$ and
    $\widetilde Y$. The variance would then be
    $\var{X \given Y} = 2 \p{1-\tau}^2 + 4\tau(1-\tau) Y^2 + 2 =
    \Order{1}$.

  Instead, we use the same samples of $Z$ when estimating both inner
  expectations. Moreover, for increased variance reduction,
    we also use an antithetic control variate based on the fact that
    $\widetilde Y$ is identically distributed to $-\widetilde Y$.  In
  summary, we set, for a given $Y$,
    \begin{equation} \label{eq:X-model} \aligned
    X &\eqdef   \frac{1}{2}\p*{P(\widetilde Y, Z) + P(-\widetilde Y, Z)}-P(Y, Z) - L_{\eta} \\
    &= \tau(Y^2 - \widetilde Y^2) + 2 \tau^{1/2}\p{1-\tau}^{1/2}\; Y Z
    - L_{\eta}.\endaligned
  \end{equation}
  Here, again, $\widetilde Y$ and $Z$ are independent standard normal
  random variables. The variance in this case is reduced to
  \[
    \sigma^2 = \var{X \given Y} = 2 \tau^2 + 4\tau(1-\tau) Y^2 =
    \Order{\tau}.
  \]
  and we can also compute analytically
  \[
    d = \abs*{\E{X \given Y}} = \abs*{\tau\p*{Y^2-1} - L_\eta}.
  \]
  \revise{The cumulative distribution function (CDF) of the random
    variable, ${\E{X \given Y}}$, is
    \begin{equation}\label{eq:model-problem-X-CDF}
    \begin{aligned}
      \prob{\E{X \given Y} \leq x} &= \prob*{ \E*{P(\widetilde Y,Z)} -
        \E*{P(Y,Z) \given Y} -
        L_{\eta} \leq x} \\
      &= \prob*{\abs{Y} \leq \p*{\frac{\tau + x +
            L_\eta}{\tau}}^{\fpow{1}{2}}} \\
      &= 1-2 \; \Phi\p*{-\p*{1 + \frac{x + L_\eta}{\tau} }^{1/2}},
    \end{aligned}
  \end{equation}
  where $\Phi$ is the standard normal CDF
  and where we substituted $\E{P(\widetilde Y,Z)} = -1$ and
  $\E{P(Y,Z) \given Y} = -\tau Y^2 - (1-\tau)$. In particular,
  \[
    \eta = \prob{\E{X\given Y} \geq 0} = 2 \; \Phi\p*{-\p*{1 +
        \frac{L_\eta}{\tau} }^{1/2}}.
  \]}
\revise{The interested reader may refer to the supplementary material
  for more motivation and discussion regarding this model
  problem. \cref{fig:model-setup-1}-(a) shows the CDF of
  ${\E{X \given Y}}$ and illustrates, as can be seen in the definition
  of the CDF in~\eqref{eq:model-problem-X-CDF}, its square-root
  behaviour in the neighbourhood of $x = -\tau - L_\eta$. This square root behaviour
  results in an inverse-square root singularity in the density of
  ${\E{X \given Y}}$ and also in $\rho$, the density of $\delta$,
  which is illustrated in \cref{fig:model-setup-1}-(b).}
\revise{
  \begin{remark}[A simple model problem]
    By relying on the one-dimensional random variable $\delta$ and its
    approximation $\hdelta$, our analysis in
    \cref{sec:nested-mc} includes cases in which $X$ is a
    function of many random variables and $Y$ is multi-dimensional,
    e.g., $X$ is the sum of losses at maturity from the options in a
    portfolio in excess of some threshold value and $Y$ is the value
    of the underlying stocks. On the other hand, our constructed
    numerical example has $Y$ being a simple one-dimensional normal
    random variable and $X$ being a simple function of $Y$ and hence
    easily sampled. This example is intentionally simple and is meant
    to showcase the advantage of combining adaptive sampling with MLMC
    without the extra complications that dealing with a large
    portfolio would entail, e.g., the cost of evaluating $X$ due to
    the large number of assets or the simulation cost of complicated
    assets.
  \end{remark}}

\def\vartau{0.02}
\def\varLc{0.080477723746297747}

\tikzset{
  declare function={
    normcdf(\x)=1/(1 + exp(-0.07056*((\x))^3 - 1.5976*(\x)));
    normpdf(\x)= exp(-\x^2/2)/ sqrt(2*pi);
    distance(\x) = \vartau * (\x^2 - 1) - \varLc;
    sigma(\x) = sqrt(2. * \vartau ^ 2 + 4 * \vartau * (1 - \vartau) * (\x^2));
    signeddelta(\x) = distance(\x) / sigma(\x);
    signeddeltaderiv(\x) = 2*\vartau*\x / sigma(\x) -
     4 * signeddelta(\x) * \x *\vartau * (1-\vartau) / (sigma(\x))^2;
  },
}

\begin{figure}[t]
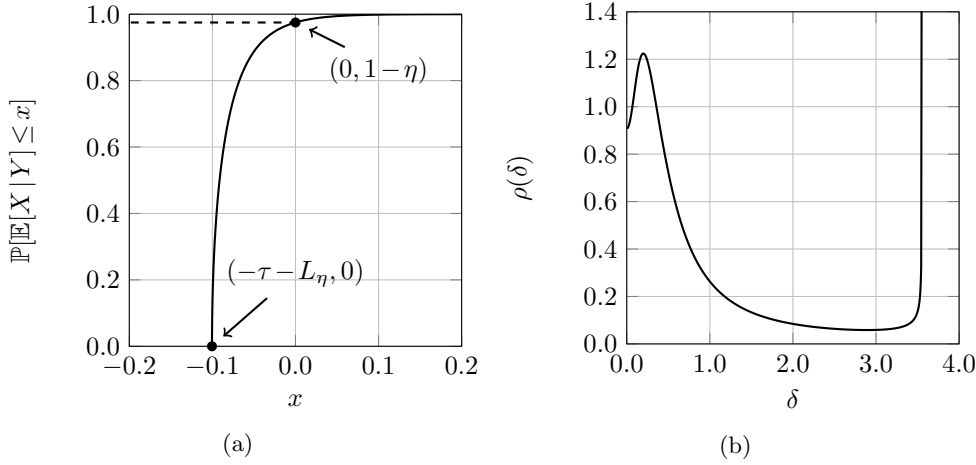

  \centering \subfig{model-problem/cdf}\hskip 0.3cm \subfig{model-problem/rho}
  \caption{\revise{(a) The cumulative distribution function of
      $\E{X \given Y}$ with $X$ as defined in~\eqref{eq:X-model} and
      $Y$ a standard normal variable. This figure illustrates the
      square-root behaviour in the neighbourhood of
      $x = -\tau-L_{\eta}$. (b) shows the density, $\rho$, of
      $\delta = {d}/{\sigma}$. This figure shows that the density is
      bounded near 0, hence \cref{ass:bounded-density} is
      satisfied, even though $\rho$ has an inverse-square root
      singularity near $\delta = 3.5$.} For both figures, we use
    $\tau=0.02$ and $L_\eta \approx 0.0805$ so that $\eta =
    0.025$. }\label{fig:model-setup-1}
\end{figure}
  \paragraph{Verifying the Assumptions}
  \revise{As \cref{fig:model-setup-1}-(b) shows, despite the
    inverse square-root singularity near $3.5$, the density $\rho$ is
    bounded near 0}, hence \cref{ass:bounded-density} is
  verified. On the other hand, to verify
  \cref{ass:bounded-moments}, we bound
    \[
    \begin{aligned}
      \kappa_q &= \sup_y \en4\{\}{\E*{\frac{\abs{X - \E{X \given Y}}^q}{\sigma^{q}}  \given Y=y}} \\
      &= \sup_{y} \frac{\E*{ \abs*{2 \tau^{1/2} \p{1\!-\!\tau}^{1/2}\;
            Z y - \tau ({\widetilde Y}^2-1)}^q}}{\p*{2
          \tau^2 + 4\tau(1-\tau) y^2}^{q/2}} \\
      &\leq \sup_{y} \frac{2^{q-1} 2^{q} \tau^{q/2}
        \p{1\!-\!\tau}^{q/2}y^q\;\E*{\abs*{ Z }^q} +
        2^{q-1}\tau^q\E*{\abs*{{\widetilde Y}^2-1}^q}}{\p*{
          4\tau(1-\tau) y^2 + 2
          \tau^2}^{q/2} } \\
      &\leq 2^{q-1}\E{\abs{Z}^q} + {2^{q/2-1}\E*{\abs*{{\widetilde
              Y}^2-1}^q}} %
      \\ &< \infty,
    \end{aligned}
  \]

  for any $q>0$.

  \paragraph{Antithetic estimators}
  The MLMC estimator in~\eqref{eq:mlmc-outer-est} uses independent
  samples of $X$, conditioned on the same value of $Y$, for the fine
  and coarse estimators of the conditional inner expectation on each
  MLMC level. We can reduce the variance of the difference by a
  constant factor by instead using an antithetic
  estimator~\cite{giles:antithetic, giles:acta}. The antithetic
  estimator uses the same set of independent samples to compute both
  the coarse and fine approximations. For a deterministic number of
  inner samples, $N_\ell = N_0 2^\ell$, we observe a variance
  reduction factor of approximately 3.5 compared to using separate
  independent samples for the fine and coarse estimators.

  To present the antithetic estimator when using an adaptive
    number of inner samples, as returned by
    \cref{alg:adaptive},
    \revise{first note that since
      \cref{alg:adaptive} returns $N_\ell$ that is a multiple
      of $N_{\ell-1}$ or vice versa, then both $\maxNell / N_\ell$ and
      $\maxNell /N_{\ell-1}$ are integers}. The antithetic estimator
    can then be written as
    \begin{equation}\label{eq:mlmc-outer-est-antithetic}
    \widehat{\eta} \eqdef
    \sum_{\ell=0}^L \frac{1}{M_\ell} \sum_{m=1}^{M_\ell} \MCest_{\ell,
      N_\ell}^{(\ell, m)}(y^{(\ell, m)}) - \MCest_{\ell,
      N_{\ell-1}}^{(\ell, m)}(y^{(\ell, m)}),
  \end{equation}
  where, the $\maxNell$ samples for the level $\ell$ correction are
  split into subsets of size $N_\ell$ and $N_{\ell-1}$ to give the
  estimator
  \[
    \MCest_{\ell, N}^{(\ell, m)} (y) \eqdef \frac{N}{\maxNell}
    \sum_{i=1}^{\maxNell/N} \heaviside{\inner^{(\ell, m)}_{N, i}(y)},
  \]
  with $\MCest_{0, N_{-1}}(\cdot) \eqdef 0$ and, similar
  to~\eqref{eq:mlmc-inner-level}, we define
  \begin{equation}
    \label{eq:MCest-antithetic}
    \inner^{\p{\ell, m}}_{N, i}(y) \eqdef \frac{1}{N} \sum_{n=1}^{N}
    x^{(\ell, m, (i-1) N + n)}(y).
  \end{equation}
  Here ${\{x^{(\ell, m, n)}(y)\}}_{n=1}^{\maxNell}$ are i.i.d.\
  samples of $X$ conditioned on $Y=y$. Moreover, these samples are
  independent from $x^{(\tilde \ell, \tilde m, n)}$ for
  $\p{\ell, m} \neq \p{\tilde \ell, \tilde m}$ and any $n$.

  \paragraph{Starting MLMC level}
  The MLMC estimator in~\eqref{eq:mlmc-outer-est-antithetic} includes
  the levels $0, 1,\ldots , L$. However, this might be sub-optimal in
  some cases. To see this, denote
  \[
    G_\ell(y) = \MCest_{\ell, N_\ell}(y) - \MCest_{\ell,
      N_{\ell-1}}(y),
  \]
  and let $V_\ell = \var{G_\ell(Y)}$ and
  $\varf_\ell = \var{\MCest_{\ell, N_\ell}(Y)}$. Moreover, let
  $W_\ell = \E{N_\ell}$ denote the expected number of inner samples,
  which, in our case, is the average work required to compute a sample
  of $G_\ell$. Recall that, to approximate a quantity of interest with
  RMS error $\tol$, the total work of MLMC is
  approximately~\cite{giles:acta}
  \[
    \tol^{-2} \p*{ \sqrt{\varf_{\ell_0} W_{\ell_0}} +
      \sum_{\ell=\ell_0+1}^L \sqrt{V_\ell W_\ell}}^2,
  \]
  for $0 \leq \ell_0 < L$. Here, we are assuming that we only include
  the levels $\ell_0, \ldots, L$ in the MLMC estimator.  Then
  including the level $\ell_0$ in the MLMC estimator is optimal if
    \begin{equation}\label{eq:mlmc-first-lvl-check}
    R_{\ell_0} \eqdef \frac{\sqrt{\varf_{\ell_0} W_{\ell_0}} +
      \sqrt{V_{\ell_0+1} W_{\ell_0+1}}}{\sqrt{\varf_{\ell_0+1}
        W_{\ell_0+1}}} \leq 1.
  \end{equation}
  Otherwise, discarding that level and starting from $\ell_0+1$ would
  yield less computational work. Note that if $\varf_{\ell} = \varf$
  and $W_{\ell} \propto 2^{\gamma\ell}$
  for all $\ell$, then the previous inequality simplifies to
  \[
    V_{\ell_0+1} %
    \leq {\p*{1-{2^{-\gamma/2}}}^2}\varf.
  \]
  For example, for $\gamma=1$, the variance of the first
  level-difference, $V_{\ell_0+1}$, to be included in the MLMC
  estimator should be more than 11 times smaller than the variance of
  the quantity of interest.

  To deal with this issue, our MLMC algorithm starts from some level,
  $\ell_0$, and approximates the variance estimates and average work
  for levels $\ell_0$ and $\ell_0+1$ using some small number of outer
  samples. Then if~\eqref{eq:mlmc-first-lvl-check} is not satisfied,
  level $\ell_0$ is discarded and the MLMC algorithm is restarted with
  the first level being $\ell_0+1$. This process is repeated
  until~\eqref{eq:mlmc-first-lvl-check} is satisfied.

  \paragraph{Results}
  We apply MLMC with deterministic and adaptive sampling.  The
  deterministic sampling algorithm is run with either
  $N_\ell = N_0 2^\ell$ or $N_\ell = N_0 4^\ell$ for $N_0=32$. On the
  other hand, the adaptive sampling algorithm is run with different
  values of $r=1.25, 1.5$ and $1.75$, the same value of $N_0$ and the
  confidence constant $C=3$. On every iteration of the adaptive
  algorithm, we use the same estimator for $d$ and $\sig$ as defined
  in~\eqref{eq:hd-estimator} and~\eqref{eq:hsig-estimator},
  respectively.  Our theory on the adaptive sampling method requires
  bounded $q$ normalised moments, $\kappa_q$, for
  $q > {2r}/\p{2-r}^2$, recall~\eqref{eq:r_adaptive_bound} in
  \cref{thm:adaptive-hd}. In our tests, we use $r=1.25, 1.5$
  and $1.75$ which requires $q > 4.45, 12$ and $56$, respectively,
  but, recall that for our model problem, $\kappa_q$ is bounded for
  all $q \geq 2$. As with \cref{fig:model-setup-1}, we set
  $\tau \eqdef 0.02$ and $L_\eta \approx 0.0805 $ so that our goal is
  to estimate $\eta = 0.025$.

  \cref{fig:model-results}-(a) shows the average number of
  used inner samples per level for the different methods. For the
  adaptive algorithm, the average number of inner samples is around 10
  times larger than $N_0 2^{\ell}$, which is used in the deterministic
  algorithm, but grows at the same rate with respect to $\ell$, as
  proved in \cref{thm:adaptive-hd}.

  On the other hand, \cref{fig:model-results}-(b) plots
  $V_\ell \eqdef \var{G_\ell}$ and
  $\varf_\ell \eqdef {\var{\MCest_{\ell, N_\ell}(Y)}}$ versus
  $\ell$. This figure shows that the variance of the MLMC levels with
  adaptive sampling is the same as the variance when using
  deterministic sampling with $N_\ell = N_0 4^\ell$, i.e., the
  variance converges like $\Order{2^{-\ell}}$ in both cases as proved
  in \cref{thm:adaptive-hd}, even though MLMC with adaptive
  sampling uses fewer inner samples per level on average. The variance
  convergence rate of the deterministic algorithm with
  $N_\ell = N_0 2^\ell$ is shown to be $\Order{2^{-\ell/2}}$, as
  proved in \cref{thm:inner-mc-var}. Note also that the
  variance of the quantity of interest, $\varf_\ell$, decreases
  slightly as $\ell$ increases but converges to the same value for all
  methods for sufficiently large $\ell$.

  \cref{fig:model-results}-(c) plots
  $E_{\ell} \eqdef \abs{\E{G_{\ell}}}$ and
  $\Ef_\ell \eqdef \abs{\E{\MCest_{\ell, N_\ell}(Y)}}$ versus
  $\ell$. The plot paints the same relative picture as the variance
  plot \cref{fig:model-results}-(b). However, recall that the
  complexity of MLMC when using adaptive sampling, is
  $\Order{\tol^{-2}\logtol2}$, does~\emph{not} depend on the
  convergence rate of $E_\ell$ since the variance, $V_\ell$, converges
  at the same rate that the average number of inner samples,
  $\E{N_\ell}$ increases. For the deterministic algorithm, on the
  other hand, since $N_\ell = N_0 2^\ell$,
  $V_\ell = \Order{2^{-\ell/2}}$ and $E_\ell = \Order{2^{-\ell}}$, the
  complexity of MLMC is $\Order{\tol^{-5/2}}$. When using
  $N_\ell = N_0 4^\ell$, MLMC has the same complexity.

  \cref{fig:model-results}-(d) plots the ratio $R_\ell$, as
  defined in~\eqref{eq:mlmc-first-lvl-check}, versus $\ell$. This
  figure shows the first level, $\ell_0$, that should be included in
  the MLMC estimator for the different methods, namely the first level
  for which $R_{\ell_0} < 1$. Hence, for the adaptive method,
  $\ell_0 = 4$ is optimal, while for the deterministic algorithm
  $\ell_0=2$ when $N_\ell = N_0 4^{\ell}$ and $\ell_0=7$ when
  $N_\ell = N_0 2^{\ell}$ are optimal.

  Finally, \cref{fig:model-results-total-work}-(a) shows the
  total number of used inner samples for the different methods for
  multiple tolerances. This number corresponds to the total work of
  each method and includes both the samples that are used to compute
  the MLMC estimator and the samples that were used in the adaptive
  algorithm (as detailed in \cref{alg:adaptive}). This
  figure, and \cref{fig:model-results-total-work}-(b) which
  shows the total running time, verify that the actual work follows
  the predicted work complexities for each of the considered
  methods. The running times of the simulations were obtained
    using a \texttt{C++} implementation of the adaptive algorithm and
    the samplers of $Y$ and $X$. Moreover, CUDA was used to
    parallelise the computation of the inner and outer samples on a
    Tesla P100 GPU with 3584 cores.

    \pgfplotstabletranspose[input colnames to={x},colnames
    from={x}]{\loadedtable}{imgs/model-heaviside/data_notdyn.dat}

\begin{figure}
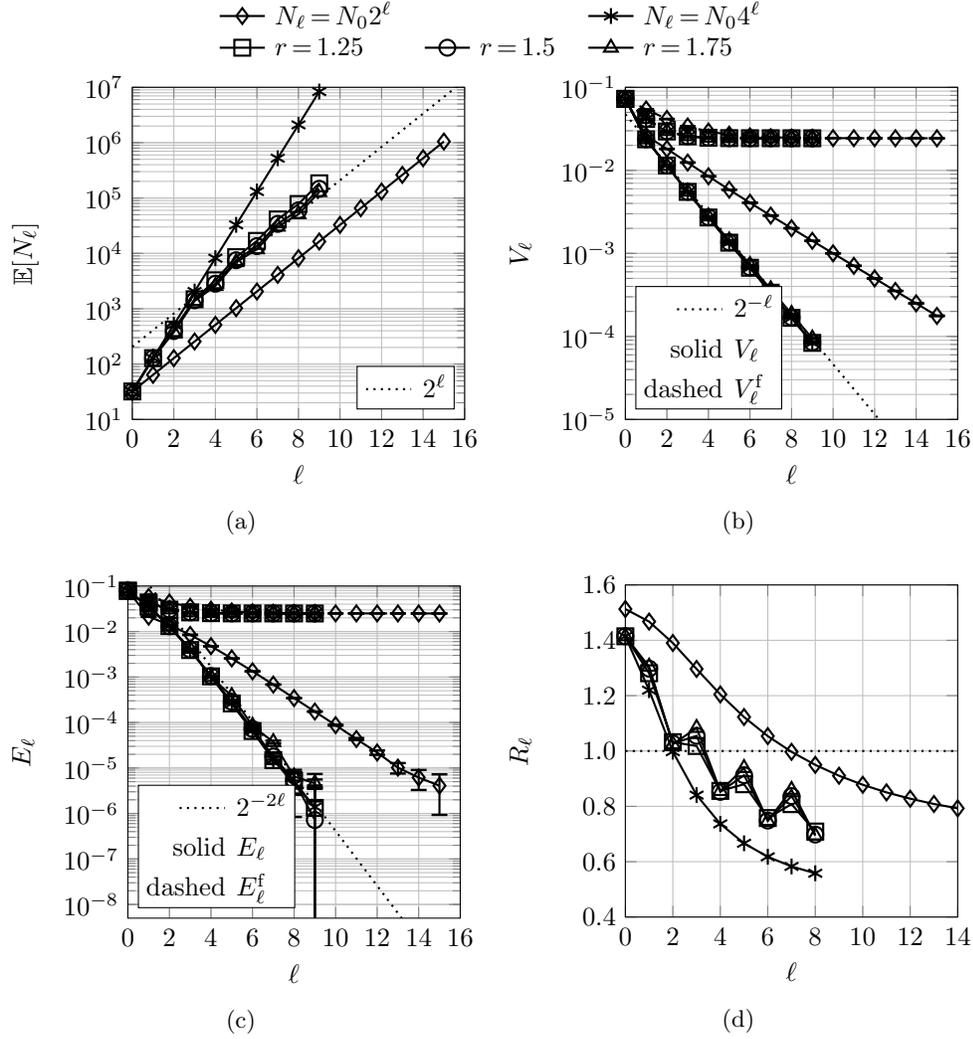

  \centering
  \begin{tabular}{lcl}
    \ref{heaviside_model_gpu/total_work-line0} {$N_\ell = N_0 2^\ell$}
    & {}
    &\ref{heaviside_model_gpu/total_work-line4} {$N_\ell = N_0
      4^\ell$}\\
    \ref{heaviside_model_gpu/total_work-line1} {$r=1.25$} & \ref{heaviside_model_gpu/total_work-line2} {$r=1.5$} &
                                                                                                               \ref{heaviside_model_gpu/total_work-line3}
                                                                                                               {$r=1.75$}
  \end{tabular}

    \subfig{model-heaviside/work}\hskip 0.3cm \subfig{model-heaviside/var}
    \subfig{model-heaviside/err}\hskip 0.3cm \subfig{model-heaviside/work-contrib}
    \caption{ (a) average number of samples, (b) variance and (c)
      absolute error per level for the MLMC estimator of
      $\E{\heaviside{\E{X \given Y}}}$ for the model problem
      described in \cref{s:model-example} using deterministic
      and adaptive sampling with different values of $r$. Note that
      the average number of samples in the adaptive method increases
      like $\Order{2^\ell}$ while the variance decreases like
      $\Order{2^{-\ell}}$.  }\label{fig:model-results}
\end{figure}

\pgfplotstabletranspose[input colnames to={x},colnames
from={x}]{\loadedtable}{imgs/model-heaviside/data.dat}

\begin{figure}
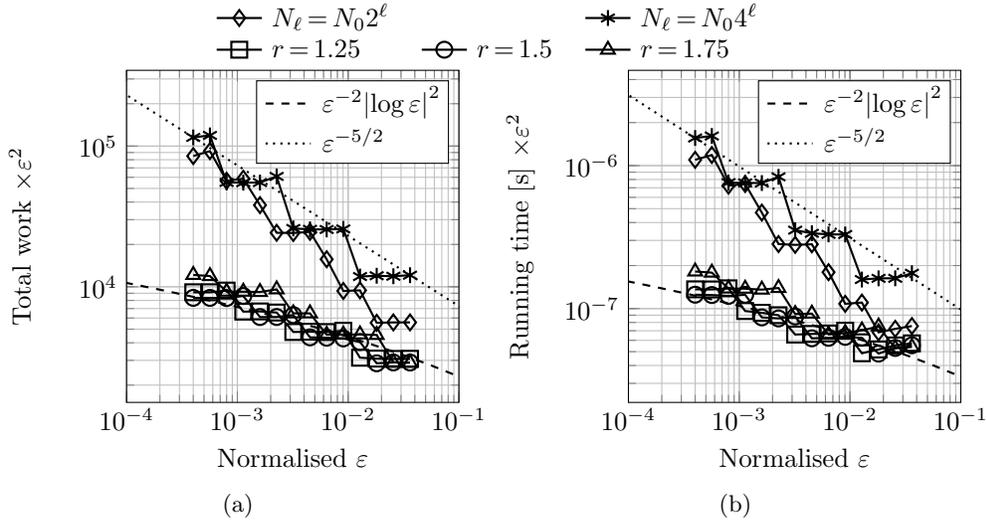

  \centering
    \begin{tabular}{lcl}
    \ref{heaviside_model_gpu/total_work-line0} {$N_\ell = N_0 2^\ell$}
    & {}
    &\ref{heaviside_model_gpu/total_work-line4} {$N_\ell = N_0
      4^\ell$}\\
    \ref{heaviside_model_gpu/total_work-line1} {$r=1.25$} & \ref{heaviside_model_gpu/total_work-line2} {$r=1.5$} &
                                                                                                               \ref{heaviside_model_gpu/total_work-line3}
                                                                                                               {$r=1.75$}
  \end{tabular}

  \subfig{model-heaviside/total_work}\hskip 0.3cm \subfig{model-heaviside/total_time}
  \caption{The (a) total work and (b) total running time for the MLMC
    estimator of $\E{\heaviside{\E{X \given Y}}}$ versus
    multiple error tolerances, normalised by the exact value,
      i.e., $\eta$, with deterministic and adaptive sampling. We also
    plot the expected work complexities in each case, as predicted in
    theory.}\label{fig:model-results-total-work}
\end{figure}

\section{Beyond Probabilities}\label{s:beyond-prob}
Using the adaptive method that we developed in the previous section
and denoting the random variable $L \eqdef {\E{X \given Y}}$, for
\[ X \eqdef \frac{1}{2}\p*{P(\widetilde Y, Z) + P(-\widetilde Y,
    Z)}-P(Y, Z), \] we can estimate
\[
  \eta = \prob{L > L_\eta} = \E{\heaviside{L - L_\eta}},
\]
for a given $L_\eta$ up to an error tolerance $\tol$ with a complexity
$\Order{\tol^{-2}\logtol2}$. Using these estimates, we can solve the
inverse problem to find $L_\eta$ for a given $\eta$, i.e., compute the
$(1 - \eta)$-quantile. In the context of financial applications,
$L_\eta$ is called the Value-at-Risk (VaR) of a financial
portfolio. Finding $L_\eta$ can be formulated as finding the root of
$\widetilde f(L_\eta) = \eta - \prob{L >L_\eta}$. To that end, we
can use a stochastic root finding algorithm such as the Stochastic
Approximation Method~\cite{robbins:SA, bardou:var-SA} and its
multilevel extensions~\cite{frikha:mlsa, dereich:gmlsa}. Instead,
since \revise{$X$ is one-dimensional} and since $\prob{L >L_\eta}$ is
monotonically decreasing with respect to $L_\eta$, we use in the
current work the simplified algorithm listed in
\cref{alg:root-finding}. The algorithm starts with an
estimate of $L_\eta$, denoted by $\hL_\eta$, then depending on where
${\heta \eqdef \prob{L >\hL_\eta}}$ lies with respect to $\eta$,
$\hL_\eta$ is adjusted. To account for the fact that $\heta$ can only
be estimated with a specified RMS error tolerance, whenever $\heta$ is
close to $\eta$, the RMS error tolerance is halved. We leave the
analysis of the root finding algorithm and comparison to other
algorithms in the literature to future work. Numerically, the
complexity of the algorithm seems to be close to
$\Order{\tol^{-2}\logtol2}$, see
\cref{fig:root-finding-complexity}.

\begin{algorithm}[t]
\caption{Stochastic root-finding algorithm.}
\label{alg:root-finding}
\KwData{$\eta$, $\eps$, $\lambda_0$, $L_0$, $h_0 > \tol/2$}
\KwResult{$\hL_\eta$ s.t.~$\abs{\hL_\eta- L_\eta} \leq \eps$}

$\hL_\eta := L_0$\;
$\lambda := \lambda_0$\;
Compute $\heta \approx \prob{L > \hL_\eta}$ with RMS error $\lambda$\;
$h := h_0\ \mbox{sign}(\heta - \eta)$\;

\vspace{0.1in}

\While{$2 \abs{h}>\eps$}{
  $\hL_\eta := \hL_\eta + h$\;
  Compute $\heta \approx \prob{L > \hL_\eta}$ with RMS error $\lambda$\;

  \If{$h\, {\rm sign}(\heta - \eta)<0$}{
    $h := - h/2$\;
  }
  \If{$\abs{\heta - \eta} < 3 \lambda$}{
    set $\lambda := \lambda/2$\;
  }
}
\Return{} $\hL_\eta$\;
\end{algorithm}

\begin{figure}
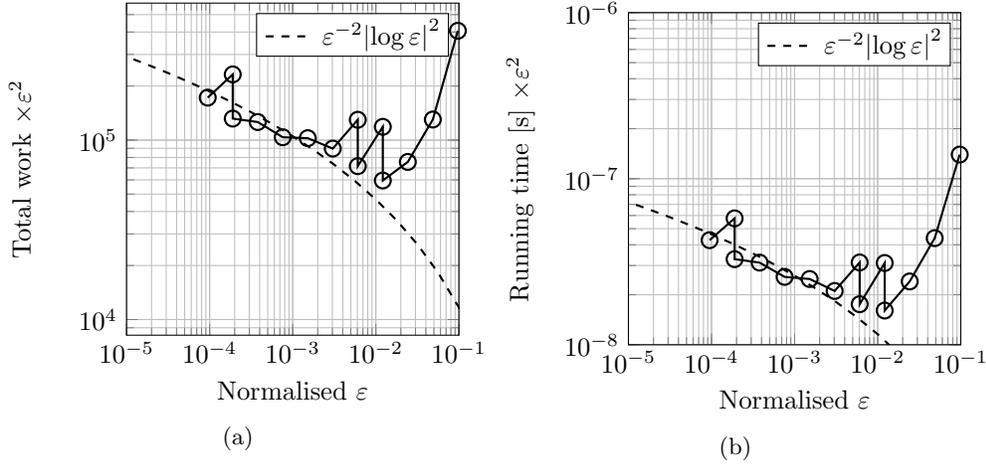

  \centering
    \subfig{root-finding/total-work}\hskip 0.3cm \subfig{root-finding/total-time}
    \caption{The complexity of \cref{alg:root-finding} to
      compute the $L_\eta$ that satisfies $\eta =\prob{L >L_\eta}$
      for a given $\eta$. In (a) the total work is the total number of
      generated samples of $X$. }\label{fig:root-finding-complexity}
\end{figure}

Another important quantity to compute is $ \E{ L \given L > L_\eta}$
for a given $\eta$. In the context of finance application, this
quantity is the Conditional Value-at-Risk (CVaR), also known as the
expected shortfall, which is the expected loss from a portfolio, given
that the loss exceeds the $(1-\eta)$-quantile for a given
$\eta$. In~\cite{rockafellar:cvaropt}, it was shown that by denoting
$f(x) \eqdef x + \frac{1}{\eta}\E{\max(L - x, 0)}$, CVaR can be written
as $\E{L \given  L > L_\eta} = f(L_\eta) = \inf_x f(x)$ since, for the
cumulative distribution function, $F(x) \eqdef \prob{L < x}$, and the
corresponding probability density function,
$\varrho \eqdef \fracs{\D F}{\D x}$, we have
\[\aligned
  \frac{\D f(x)}{\D x} &= 1 -
  \frac{1}{\eta}\E{\heaviside{L-x}} = 1-\frac{1-F\p{x}}{\eta},\\
  \frac{\D^2 f(x)}{\D x^2} &= \frac{\D}{\D
    F(x)}\p*{1-\frac{1-F(x)}{\eta}} \frac{\D F(x)}{\D x} =
  \frac{1}{\eta} \varrho(x) \geq 0, \endaligned
\]
and $\fracs{\D f(x)}{\D x} \big|_{x=L_\eta} = 0$.  Hence if $L_\eta$
is approximated with $\hL_\eta$, using, for example,
\cref{alg:root-finding}, and the CVaR value, $f(L_\eta)$, is
approximated with $f(\hL_\eta)$, then the error is
\[
  \abs{f(L_\eta) - f(\hL_\eta)} = \Order*{\p*{L_\eta - \hL_\eta}^2}.
\]
That is, an $\Order{\tol^{1/2}}$ error in the approximation of the
VaR, $L_\eta$, with \linebreak $\Order*{\tol^{-1}\:\logtol2}$ complexity yields
an $\Order{\tol}$ error in the approximation of the CVaR
$f(L_\eta) = {\E{L \given L > L_\eta}}$.

To approximate CVaR given $\hL_\eta$ by computing $f(\hL_\eta)$, we
still need to approximate the expectation
$\E*{\max\p*{\E{X \given Y} - \hL_\eta, 0}}$ where the outer
expectation is with respect to $Y$ while the inner conditional
expectation is with respect to $X$.
Without loss of generality, we can set $\hL_\eta \eqdef 0$ by defining
$X_{\text{new}} \eqdef X_{\text{old}} - \hL_\eta$. The resulting
problem, to compute $\E*{\max\p*{\E{X \given Y}, 0}}$,
is similar to~\eqref{eq:objective-nested} but with a maximum function
instead of a step function. Hence, we can again use the MLMC method,
as described \cref{s:mlmc-nested}, with an antithetic sampler,
as mentioned in \cref{rem:antithetic} and explained in
\cref{s:model-example}. Using the notation in
\cref{s:model-example}, we set
    \begin{equation}\label{eq:max-anti-estimator}
    \MCest_{\ell, N}^{(\ell, m)} (y) \eqdef \frac{N}{\maxNell}
    \sum_{i=1}^{\maxNell/N} \max\p*{\inner^{(\ell, m)}_{N, i}(y), 0}.
  \end{equation}
  Note that, for a given
  $y$, whenever the estimates $\inner^{(\ell,m)}_{N_\ell,
    i}(y)$ and $\inner^{(\ell, m)}_{N_{\ell-1},
    i}(y)$ are \revise{positive for all $i$, then (cf.~\eqref{eq:MCest-antithetic})
  \[
    \MCest_{\ell, N_{\ell}}^{(\ell, m)} (y) = \frac{1}{\maxNell}
    \sum_{n=1}^{\maxNell} x^{(\ell, m, n)}(y) = \MCest_{\ell,
      N_{\ell-1}}^{(\ell, m)} (y)
  \]
  and hence the difference $\MCest_{\ell, N_{\ell}}^{(\ell, m)} (y) -
  \MCest_{\ell, N_{\ell-1}}^{(\ell, m)}
  (y)$ is zero. Similarly, if $\inner^{(\ell,m)}_{N_\ell,
    i}(y)$ and $\inner^{(\ell, m)}_{N_{\ell-1},
    i}(y)$ are negative for all
  $i$, then the difference is trivially zero}. Using a deterministic
number of samples $N_\ell = N_0
2^{\ell}$ results in an error
$\Order{2^{-\ell/2}}$ in estimating the inner expectation, $\E{X
  \given
  Y}$. \revise{Therefore, the two estimates,
  $\inner^{(\ell,m)}_{N_\ell, i}(y)$ and $\inner^{(\ell,
    m)}_{N_{\ell-1}, i}(y)$,
might have different signs whenever the exact value}, ${\E{X \given
  Y}}$, is
$\Order{2^{-\ell/2}}$. Hence the variance of the MLMC difference is
$\Order{2^{-3
    \ell/2}}$ and its absolute expectation is
$\Order{2^{-\ell}}$. This is made more precise in the following
theorem.

\begin{theorem}\label{thm:max-conv-rates}
  Assume that \cref{ass:bounded-density}
  and~\ref{ass:bounded-moments} hold and assume further that
  $\E{\sigma^2} < \infty$ and that there exists
  $\sigma_0$ such that $\sigma \leq \sigma_0$ given $\delta <
  \delta_0$. Then denoting
  \[
    G_\ell \eqdef \MCest_{\ell, N_\ell} (Y) - \MCest_{\ell,
      N_{\ell-1}} (Y),
  \]
  where $\MCest_{\ell,
    N}$ is defined in~\eqref{eq:max-anti-estimator} and $N_\ell = N_0
  2^{\ell}$, we have
  \[
    \E{\abs{G_\ell}} = \Order{N_{\ell}^{-1}} \quad\text{and}\quad
    \E{G_\ell^2} = \Order{N_{\ell}^{{-\min\p{3, q}}/{2}} }.
  \]
\end{theorem}
\noindent \emph{Comment.} The proof follows the same ideas as
  in~\cite[Theorem 5.2]{giles:antithetic} where a similar result was
  derived for an antithetic estimator with respect to
  time-discretisation. The same proof idea was also employed
  in~\cite[Theorem 2.3]{bujok:bernoulli} where the result was shown
  for the same antithetic estimator considered here, a more generic
  piece-wise linear function, and
  $X$ being a Bernouli random variable.
\begin{proof}
  Since $N_\ell = 2N_{\ell-1}$, the antithetic estimator is
  \[\aligned
    G_\ell &= \max\p*{0, \frac{1}{2}\sum_{i=1}^2\inner_{N_{\ell-1}, i}(Y)} - \frac{1}{2}
    \sum_{i=1}^{2} \max\p*{0,
      \inner_{N_{\ell-1}, i}\p{Y}} %
    \endaligned
  \]
  For a given $Y$, define the Bernoulli random variable
  \revise{\[
    B = \ind{ \abs*{\inner_{N_{\ell-1}, 1}(Y) - \E{X
          \given Y}} \geq \abs{\E{X \given Y}}\;\vee\;
    \abs*{\inner_{N_{\ell-1}, 2}(Y) - \E{X
          \given Y}} \geq \abs{\E{X \given Y}}},
  \]
  i.e, $B=1$ whenever $\abs*{\inner_{N_{\ell-1}, i}(Y) - \E{X \given
      Y}} \geq \abs{\E{X \given Y}}$ for $i=1$ or
  $2$ and zero otherwise.
  Then for $p=1$ or $2$,
  \[
    \E{\abs{G_\ell}^{p}} = \E{\abs{G_\ell}^{p} {B}} + \E{\abs{G_\ell}^{p} \p{1-B}}.
  \]
  Considering the second term, when $B =
  0$, we have that both $\inner_{N_{\ell-1},
    1}$ and $\inner_{N_{\ell-1},
    2}$ share the same sign. Therefore, in this case, $G_\ell =
  0$ and the expected value is zero. On the other hand, for the first
  term, using H\"older's inequality gives
    \begin{equation}\label{eq:G-ell-sq-holder}
    \E{G_\ell^p B} = \E[\Big]{\E*{G_\ell^p B \given Y}}
    \leq \E*{\E*{\abs{G_\ell}^{q}\given Y}^{p/q} \  \E{B \given
        Y}^{1-p/q}},
  \end{equation}
  where using~\eqref{eq:tail-probability-bound} yields
  \[\aligned
    \E{B \given Y} &\leq \sum_{i=1}^{2} \prob{\abs{\inner_{N_{\ell-1},
        i}(Y) - \E{X
          \given Y}} \geq \abs{\E{X \given Y}}} \\
    &\leq 2 \min\p*{1, C_q \kappa_q \delta^{-q} N_{\ell-1}^{-q/2}}. \\
    \endaligned
  \]
  On the other hand, since $2\max\p*{0, x} = x+\abs{x}$, we have that
  \[\aligned
    f(x_1,x_2) &\eqdef \max\p*{0, \frac{x_1+x_2}{2}} - \frac{1}{2}
    \max\p*{0, x_1} - \frac{1}{2}
    \max\p*{0, x_2}\\
    &= \frac{1}{4}\p*{\abs*{x_1+x_2} - \abs*{x_1} - \abs*{x_2}}
    \endaligned
  \]
  and $f(x_1, x_2) \leq 0$ by the triangular inequality while
  \[
    \begin{aligned}
      4 f(x_1, x_2) &= \abs{2x_1-\p{x_1-x_2}} - \abs{x_1} - \abs{x_1 +
        (x_2-x_1)} \\
      &\geq 2 \abs{x_1} - \abs{x_1 - x_2} - \abs{x_1} - \abs{x_1}
      - \abs{x_2 - x_1}\\
      &= - 2 \abs{x_1 - x_2}.
    \end{aligned}
  \]
  Therefore
  \[
    \abs{f(x_1, x_2)} \leq \frac{1}{2} \abs{x_1-x_2} \leq
    \frac{1}{2}\p[\Big]{\abs*{x_1-x} + \abs*{x_2-x}}
  \]
  for any $x$. Using this and Jensen's inequality yields
  \[\aligned
    \E*{\abs*{G_\ell}^{q} \given Y}^{p/q} &\leq
    \E*{2^{-q} \p*{\sum_{i=1}^2\abs*{\inner_{N_{\ell-1, i}}(Y)-\E{X \given
            Y}} }^{q} \given Y}^{p/q} \\
    &\leq 2^{-p/q}\sum_{i=1}^{2} \E*{{\abs*{\inner_{N_{\ell-1},i}(Y)-\E{X \given Y}}
      }^{q} \given Y}^{p/q}. \endaligned
  \]
  Finally, using \cref{lemma:mc-prob-bound}, yields
  \[\aligned
    \E*{\abs*{G_\ell}^{q} \given Y}^{p/q} &\leq 2^{1-p/q} C_{q}^{p/q}
    \kappa_{q}^{p/q} \sigma^{p} N_{\ell-1}^{-p/2} = 4^{1-1/q}
    C_{q}^{p/q} \kappa_{q}^{p/q} \sigma^{p} N_\ell^{-p/2}. \endaligned
  \]
  Substituting back into~\eqref{eq:G-ell-sq-holder},
  there exist constants $c_1, c_2$ and $q'=q(1-p/q)$, independent of
  $\ell$, such that
  \[\aligned
    \E{\E{G_\ell^p B \given Y}} &\leq c_1 \E*{\sigma^{p}
      N_\ell^{-p/2} \min\p*{1, c_2 \delta^{-q'} N_{\ell}^{-q'/2}}} \\
    &\leq c_1 N_\ell^{-p/2} \E*{\sigma^{p} \min\p*{1, c_2 \delta^{-q'}
        N_{\ell}^{-q'/2}}\: \ind{\delta < \delta_0}} \\
    & \hskip 1cm + c_1 N_\ell^{-p/2} \E*{\sigma^{p} \min\p*{1, c_2
        \delta^{-q'}
        N_{\ell}^{-q'/2}}  \: \ind{\delta > \delta_0}} \\
    &\leq c_1 N_\ell^{-p/2} \sigma_0^{p} \E*{ \min\p*{1, c_2 \delta^{-q'}
        N_{\ell}^{-q'/2}} \: \ind{\delta < \delta_0}}  \\
    & \hskip 1cm + c_1
    c_2 \delta_0^{-q'} N_{\ell}^{-1-q'/2} \E*{\sigma^{p} \: \ind{\delta > \delta_0}} \\
    &\leq c_1N_\ell^{-p/2} \rho_0\sigma_0^p \int_0^\infty \min\p*{1, c_2
      \delta^{-q'} N_{\ell}^{-q'/2}} \D \delta  \\
    & \hskip 1cm + c_1 c_2 \delta_0^{-q'}
    N_{\ell}^{-1-q'/2} \E{\sigma^p} \\
    &\leq \frac{q'}{q'-1} c_1 c_2^{1/q'}\rho_0 \sigma^p_0 N_\ell^{-(p+1)/2} +
    c_1 c_2 \delta_0^{-q'} N_{\ell}^{-1-q'/2} \E{\sigma^p}.
    \endaligned
  \]
  Substituting $q'$ and assuming that
  $\E{\sigma^p}$ is bounded, yields
  \[
    \E{G_\ell^p B} \leq \Order{N^{-\min\p*{p+1, 2 - p + q}/2}}.
  \]
  which gives the two results for $p=1$ or $2$ and $q > 2$.}
\end{proof}

Hence, for $q > 2$, the complexity of the resulting MLMC method is
then $\Order{\tol^{-2}}$ since ${\alpha=1}$ and ${\beta > \gamma = 1}$
using the notation of~\cite{giles:acta}. This complexity was
  also shown in~\cite{bujok:bernoulli}. This is also the optimal
complexity of MLMC and therefore using adaptive sampling does not
improve the complexity.  Nevertheless, we can use exactly the same
adaptive \cref{alg:adaptive} to select a random number of
samples, $N_\ell$, depending on $Y$. We omit the analysis of the
resulting MLMC method and refer instead to the numerical results in
\cref{fig:model-max-results}. This figure shows that the
variance and absolute errors converge like $\Order{2^{-3\ell}}$ and
$\Order{2^{-2\ell}}$, respectively, when using an adaptive number of
samples. Even though this is faster than $\Order{2^{-3\ell/2}}$ and
$\Order{2^{-\ell}}$, respectively, that are observed when using a
deterministic number of samples (and proved in
\cref{thm:max-conv-rates}), the overall complexity of the MLMC
is the same when using the two sampling methods, as shown in
\cref{fig:model-max-results}-(d).

\pgfplotstabletranspose[input colnames
to={x}, colnames from={x}]{\loadedtable}{imgs/model-max/data.dat}

\begin{figure}
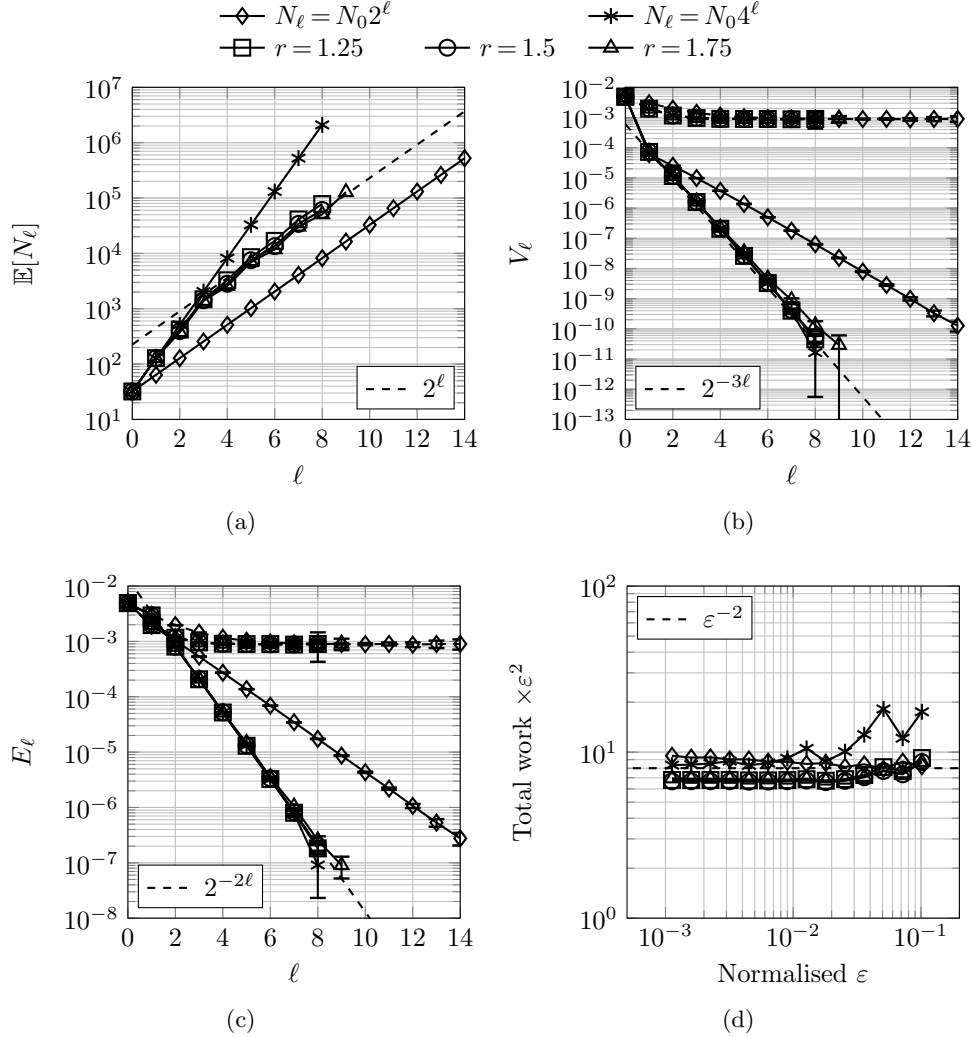

  \centering
  \begin{tabular}{lcl}
    \ref{max_model_gpu/total_work-line0} {$N_\ell = N_0 2^\ell$}
    & {}
    &\ref{max_model_gpu/total_work-line4} {$N_\ell = N_0
      4^\ell$}\\
    \ref{max_model_gpu/total_work-line1} {$r=1.25$} & \ref{max_model_gpu/total_work-line2} {$r=1.5$} &
                                                                                                           \ref{max_model_gpu/total_work-line3}
                                                                                                           {$r=1.75$}
  \end{tabular}

    \subfig{model-max/work}\hskip 0.3cm \subfig{model-max/var}
    \subfig{model-max/err}\hskip 0.3cm \subfig{model-max/total_work}
    \caption{ (a) average number of samples, (b) variance and (c)
      absolute error per level for the MLMC estimator of
      $\E{\max(\E{X\given Y, 0})}$ using antithetic deterministic and
      adaptive sampling with different values of $r$.  Also, (d) shows
      the total work of the MLMC estimator. We also plot the predicted
      rates, as discussed in \cref{s:beyond-prob}. }\label{fig:model-max-results}
\end{figure}

\section{Conclusions}\label{sec:conc}
In this work, we presented a MLMC method for nested expectations with
a step function, in which deeper levels use more samples for a Monte
Carlo estimator of the inner conditional expectation. We also
presented an adaptive algorithm that selects the number of inner
samples given a sample of the outer random variable. We showed that
under certain assumptions, the variance of the MLMC levels decreases
at the same rate that the work increases and hence the MLMC method
achieves a near-optimal $\Order{\tol^{-2} \logtol2}$ complexity for a
RMS error tolerance $\tol$. We also showed how our methods can be
combined with a root-finding algorithm to compute more complicated
risk measures, namely, the Value-at-Risk (VaR) and the Conditional
Value-at-Risk (CVaR) with the latter being obtained with
$\Order{\tol^{-2}}$ complexity.

The next step in our our work is to apply MLMC with adaptive sampling
to the problem of estimating the probability of large loss from a
financial portfolio consisting of many financial options based on
underlying assets described by general SDEs. Using unbiased MLMC, this
probability can be written as a nested expectation even in the case
when paths of the underlying stochastic differential equation must be
estimated using a time-stepping scheme. Moreover, various control
variates and sampling strategies make computing this probability more
efficient. The result is that the complexity of computing CVaR is also
$\Order{\tol^{-2}}$, independent of the number of options in the
portfolio.

\bibliographystyle{siamplain}

\begin{thebibliography}{10}

\bibitem{bardou:var-SA}
{\sc O.~Bardou, N.~Frikha, and G.~Pages}, {\em Computing {VaR} and {CVaR} using
  stochastic approximation and adaptive unconstrained importance sampling},
  Monte Carlo Methods and Applications, 15 (2009), pp.~173--210.

\bibitem{broadie:adapt}
{\sc M.~Broadie, Y.~Du, and C.~C. Moallemi}, {\em Efficient risk estimation via
  nested sequential simulation}, Management Science, 57 (2011), pp.~1172--1194.

\bibitem{bujok:bernoulli}
{\sc K.~Bujok, B.~Hambly, and C.~Reisinger}, {\em Multilevel simulation of
  functionals of {B}ernoulli random variables with application to basket credit
  derivatives}, Methodology and Computing in Applied Probability, 17 (2015),
  pp.~579--604.

\bibitem{burkholder:1966}
{\sc D.~L. Burkholder}, {\em Martingale transforms}, The Annals of Mathematical
  Statistics, 37 (1966), pp.~1494--1504.

\bibitem{dereich:gmlsa}
{\sc S.~Dereich and T.~Mueller-Gronbach}, {\em General multilevel adaptations
  for stochastic approximation algorithms}, arXiv preprint arXiv:1506.05482,
  (2015).

\bibitem{frikha:mlsa}
{\sc N.~Frikha et~al.}, {\em Multi-level stochastic approximation algorithms},
  The Annals of Applied Probability, 26 (2016), pp.~933--985.

\bibitem{giles:MLMC}
{\sc M.~B. Giles}, {\em Multilevel {M}onte {C}arlo path simulation}, Operations
  Research, 56 (2008), pp.~607--617.

\bibitem{giles:acta}
{\sc M.~B. Giles}, {\em Multilevel {M}onte {C}arlo methods}, Acta Numerica, 24
  (2015), pp.~259--328.

\bibitem{giles:antithetic}
{\sc M.~B. Giles and L.~Szpruch}, {\em Antithetic multilevel {M}onte {C}arlo
  estimation for multi-dimensional {SDE}s without {L}\'evy area simulation},
  The Annals of Applied Probability, 24 (2014), pp.~1585--1620,
  \url{https://doi.org/10.1214/13-AAP957}.

\bibitem{giorgi:ml2r}
{\sc D.~Giorgi, V.~Lemaire, et~al.}, {\em Limit theorems for weighted and
  regular multilevel estimators}, Monte Carlo Methods and Applications, 23
  (2017), pp.~43--70.

\bibitem{gordy:nested}
{\sc M.~B. Gordy and S.~Juneja}, {\em Nested simulation in portfolio risk
  measurement}, Management Science, 56 (2010), pp.~1833--1848.

\bibitem{gou:ms-var}
{\sc W.~Gou}, {\em Estimating {V}alue-at-{R}isk using {M}ultilevel {M}onte
  {C}arlo {M}aximum {E}ntropy method}, master's thesis, University of Oxford,
  2016.

\bibitem{hajiali:msthesis}
{\sc A.-L. Haji-Ali}, {\em Pedestrian flow in the mean-field limit}, master's
  thesis, King Abdullah University of Science and Technology, 2012.
\newblock MSc\ thesis, KAUST.

\bibitem{hong:var-review}
{\sc L.~J. Hong, Z.~Hu, and G.~Liu}, {\em {M}onte {C}arlo methods for
  {V}alue-at-{R}isk and {C}onditional {V}alue-at-{R}isk: a review}, ACM
  Transactions on Modeling and Computer Simulation (TOMACS), 24 (2014), p.~22.

\bibitem{robbins:SA}
{\sc H.~Robbins and S.~Monro}, {\em A stochastic approximation method}, The
  Annals of Mathematical Statistics,  (1951), pp.~400--407.

\bibitem{rockafellar:cvaropt}
{\sc R.~T. Rockafellar and S.~Uryasev}, {\em Optimization of conditional
  {V}alue-at-{R}isk}, Journal of Risk, 2 (2000), pp.~21--42.

\end{thebibliography}
 \end{document}